\documentclass[format=acmsmall, review=false]{acmart}

\usepackage{acm-ec-17}

\usepackage[ruled]{algorithm2e}

\usepackage{amsmath,amsfonts,amssymb,amsthm}
\SetArgSty{textrm}  
\SetAlFnt{\small}
\SetAlCapFnt{\small}
\SetAlCapNameFnt{\small}
\SetAlCapHSkip{0pt}
\IncMargin{-\parindent}

\usepackage{breakcites}
\usepackage{microtype}
\usepackage{hyperref}
\usepackage{graphicx}

\newcommand{\rn}[1]{{\textcolor{black}{#1}}}

\newcommand{\lin}{{{\mbox{\bf Lin}}}}
\newcommand{\leon}{{{\mbox{\bf Leo}}}}

\newcommand{\ssigma}{\mbox{\boldmath $\sigma$}}
\newcommand{\pp}{{\mbox{\boldmath $p$}}}
\newcommand{\vv}{{\mbox{\boldmath $v$}}}
\newcommand{\bb}{{\mbox{\boldmath $\bid$}}}
\newcommand{\ttt}{{\mbox{\boldmath $t$}}}
\newcommand{\xx}{\mbox{\boldmath $x$}}
\newcommand{\zz}{\mbox{\boldmath $z$}}
\newcommand{\yy}{\mbox{\boldmath $y$}}

\newcommand{\sss}{\mbox{\boldmath $s$}}
\newcommand{\CA}{\mbox{${\mathcal A}$}}

\newcommand{\CM}{\mbox{${\mathcal M}$}}
\newcommand{\CL}{\mbox{${\mathcal L}$}}
\newcommand{\smi}{{\mbox{${\ssigma_{-i}}$}}}
\newcommand{\Real}{\mbox{${\mathbb R}$}}
\newcommand{\Rplus}{\mbox{${\mathbb R_+}$}}

\newtheorem{claim}{Claim}

\newcommand{\cal}{\mathcal}

\newcommand{\thmref}[1]{Theorem~\ref{#1}}
\newcommand{\lemref}[1]{Lemma~\ref{#1}}

\renewcommand{\vec}[1]{\mathbf{#1}}

\newcommand{\budget}{B}

\newcommand{\totbudget}{{\cal{B}}}
\newcommand{\bid}{b}
\newcommand{\eqx}{\tilde{\vec{x}}}
\newcommand{\eqp}{\tilde{p}}
\newcommand{\eqbb}{\tilde{\bb}}
\newcommand{\eqbid}{\tilde{\bid}}
\newcommand{\optx}{\vec{x}^*}
\newcommand{\devx}{\vec{x}'}

\newcommand{\eps}{\epsilon}

\newcommand{\simina}[1]{{\color{red}\noindent\textbf{SIMINA says }\marginpar{****}\textit{{#1}}}}

\newcommand{\vasilis}[1]{}
\newcommand{\ruta}[1]{{\color{blue}\noindent\textbf{RUTA says }\marginpar{****}\textit{{#1}}}}
\begin{document}

\title[Nash Social Welfare Approximation for Strategic Agents]{Nash Social Welfare Approximation for Strategic Agents}


\author{Simina Br\^{a}nzei}
\affiliation{
	\institution{Hebrew University of Jerusalem}
	\department{Computer Science and Engineering}
	\city{Jerusalem}
	\country{Israel}
}

\author{Vasilis Gkatzelis}
\affiliation{%
	\institution{Drexel University}
	\department{Computer Science}
	\streetaddress{}
	\city{Philadelphia}
	\state{PA}
	\country{USA}}
\author{Ruta Mehta}
\affiliation{%
	\institution{University of Illinois, Urbana-Champaign}
	\department{Computer Science}
	\city{Urbana}
	\state{IL}
	\country{USA}
}

\begin{abstract}
	The fair division of resources among strategic agents is an important age-old problem that has led to a rich body of literature. At the
	center of this literature lies the question of whether there exist mechanisms that can implement fair outcomes, despite the agents'
	strategic behavior. A fundamental objective function used for measuring the fairness of an allocation is the geometric mean of the 
	agents' values, known as the {\em Nash social welfare} (NSW). This objective function is maximized by
	widely known solution concepts such as Nash bargaining and the competitive equilibrium with equal incomes. 
	
	In this work we focus on the question of (approximately) implementing this objective. The starting point of our analysis is the Fisher
	market, a fundamental model of an economy, whose benchmark is precisely the (weighted) Nash social welfare. We begin by studying two extreme
	classes of valuations functions, namely perfect substitutes and perfect complements, and find that for perfect substitutes, the Fisher
	market mechanism yields a constant approximation: at most 2 and at least $e^{\frac{1}{e}}$ ($\approx$ 1.44). However, for perfect
	complements, the Fisher market mechanism does not work well, its bound degrading linearly with the number of players.
	
	Strikingly, the Trading Post mechanism---an indirect market mechanism also known as the Shapley-Shubik game---has significantly better
	performance than the Fisher market on its own benchmark. Not only does Trading Post achieve an approximation of 2 for perfect substitutes, but this bound holds for any concave utilities, and it becomes essentially optimal for perfect complements, where it reaches $(1+\epsilon)$ for
	any $\epsilon>0$. Moreover, we show that all the Nash equilibria of the Trading Post mechanism are pure (hence the approximation factors extend to all Nash equilibria), and satisfy an important notion of individual fairness known as proportionality. 
\end{abstract}

\maketitle

\section{Introduction}

The question of allocating resources among multiple participants in a way that is fair is as old as human society itself \cite{Moulin03}, with some of
the earliest recorded instances dating back to more than 2500 years ago\footnote{See, e.g., Hesiod's Theogony, 
	where a protocol known as \emph{Cut and Choose} is mentioned.}. The mathematical study of fair division began with the work of 
Steinhaus during the second world war, which led to an extensive and growing body of work on fair division protocols 
within economics and political science, e.g., \cite{Young94,BramsT96,RobertsonWebb98,Moulin03,Barbanel04}. Recent years have
seen an increased amount of work on fair division coming from computer science (see, e.g., \cite[Part II]{COMSOC}), partly 
motivated by problems related to allocating computational resources---such as CPU, memory, and bandwidth---among the users of 
a computing system. This work has focused on settings with both divisible goods (e.g., \cite{CLPP10, Pro13, CGG13a, BM15, PPS12,BCKP16}) 
and indivisible ones (e.g., Procaccia and Wang \citeyear{PW14}, \citet*{CCK09}, \citet*{CG15}, Aziz and Mackenzie \citeyear{AM16}).


One of the basic questions underlying the fair division problem is that of \emph{defining fairness} 
to begin with, and a large body of work in economics, particularly social choice theory, is concerned with this 
very question, with numerous solution concepts proposed in response. Our fairness concept of choice 
herein is the \emph{Nash social welfare} (NSW), which dates back to the fifties~\cite{Nash50} (also \cite{KN79}) and
has been proposed by Nash as a solution for bargaining problems, using an axiomatic approach. This objective 
aims to choose an outcome $\vec{x}$ maximizing the geometric mean of the utilities ($u_i(\vec{x})$) of the 
$n$ participating agents and, like other standard welfare objectives, it is captured by a family 
of functions known as generalized (power) means:
\begin{small}
	\[ M_p(\vec{x}) ~=~ \left(\frac{1}{n}\sum_i \left[u_i(\vec{x})\right]^p\right)^{1/p}.\]
\end{small}
In particular, the NSW corresponds to $M_0(\vec{x})$, the limit of $M_p(\vec{x})$ as $p$ goes to zero, i.e., $\left( \Pi_i
u_i(\vec{x})\right)^\frac{1}{n}$. 

While an extended treatment of the NSW can be found, for example, in~\cite{Moulin03}, we highlight a fundamental property of the NSW
objective, namely that it achieves \emph{a natural compromise between individual fairness and efficiency}.  Two other well-studied
functions captured by $M_p(\vec{x})$ are the $(i)$ {\em egalitarian} (max-min) objective attained as $p\rightarrow -\infty$, and $(ii)$
{\em utilitarian} (average) objective attained at $p=1$, which correspond to extreme fairness and extreme efficiency, respectively.
However, the former may cause vast inefficiencies, while the latter can completely neglect how unhappy some agents might be. The NSW
objective lies between these two extremes and strikes a natural balance between them, since maximizing the geometric mean leads to more
balanced valuations, but without neglecting efficiency.

The highly desired fairness and efficiency trade-off that the NSW objective provides can be verified via
its close connection with market equilibrium outcomes in the \emph{Fisher market} model---one of the fundamental 
resource allocation models in mathematical economics. This model was developed by Fisher~\cite{BSAD} 
and studied in an extensive body of literature~\cite{gale,EG,Ga76,Ea,DPSV,orlin,CDSYao04,BrunoVICALP04,homothetic,auction.gross}. 
The basic setting involves a seller who brings multiple divisible goods to the market and a set of buyers equipped  
with monetary endowments (budgets). The goal of the seller is to extract as much money as possible from the buyers, by charging money for the goods (through prices), while each buyer aims to acquire the best possible bundle of goods at the given prices.
A {\em market equilibrium} is an outcome where supply meets demand, and has been shown to exist for very general models of an economy. When the buyers have the same budgets, this outcome is known as a \emph{competitive equilibrium
	from equal incomes} (CEEI)~\cite{Varian74}. For a broad family of valuations, including the ones considered in this 
paper, the market equilibrium allocation is known to maximize the NSW objective when the budgets are equal. In other 
words, the seller's goal can be achieved by computing the allocation of the goods that maximizes the NSW which, 
in turn, implies the desired price for each good. 

A fundamental problem in implementing the Nash social welfare objective is informational: the seller needs to know the valuations of the participants.
When these valuations are private information of the buyers, a natural candidate is the mechanism induced by the Fisher market, known as the \emph{Fisher
	market mechanism}: ask the buyers to report their valuations, and then compute the NSW maximizing allocation based on
the reports. Unfortunately, it is well known that buyers can feign different interests and eventually 
get better allocations~\cite{ABGMS,CDZ11,CDZZ12,BBNR,BCDFFZ14}. This strategic behavior can result in the mechanism computing 
market equilibria with respect to preferences having little to do with reality, leading
to unfair allocations. 
In this work we address the following basic question:
\medskip
\begin{quote}
	\emph{How well does the Fisher market mechanism optimize its own objective---the Nash social welfare---when the participants 
		are strategic? Are there better mechanisms?}
\end{quote}
\medskip

This question falls under the general umbrella of implementation theory \cite{DHM}, and particularly, of implementing markets
\cite{SN,BCW,DS,GG,PS}. In this literature, the goal is to
identify mechanisms (game forms) for which the set of Nash equilibria coincides with the set of market equilibrium
allocations for every possible state of the world \cite{DHM}. In general this can be achieved only in the limit, as 
the number of players goes to infinity and every player is infinitesimal compared to the entire economy~\cite{DG}. 

In this work we show that even for small markets there exist classic mechanisms achieving outcomes that
closely approximate the optimal NSW on every instance while simultaneously guaranteeing individual fairness.
We measure the quality of a mechanism using the \emph{Price of Anarchy} (PoA)~\cite{NRTV07}, defined as the ratio 
between the optimal NSW 
and the NSW of the worst Nash equilibrium outcome obtained by the mechanism.

\subsection{Our Results}

We study the question of approximately implementing the NSW objective 
starting with two extensively studied classes of valuations, namely linear (or additive) and Leontief (\cite{NRTV07}, \cite{gale}, \cite{BrunoVICALP04}), two
extremes, and generalize many of our results to arbitrary concave valuations. Recall that additive valuations capture goods that are \emph{perfect substitutes}, i.e., that can replace each other in consumption, such
as Pepsi and Coca-Cola. Leontief valuations capture \emph{perfect complements}, i.e., goods that have no value without each other, such
as left and right shoes. 

In addition to the aggregate measure of wealth captured through the Nash welfare, we will analyze mechanisms that guarantee \emph{proportionality}, one of the fundamental fairness notions in fair division. Proportionality requires that every agent $i$ gets at least a fraction $B_i/\mathcal{B}$ of its utility for everything, where $B_i$ is its budget and $\mathcal{B}$ the sum of all budgets.

Our first set of results concerns the Fisher market mechanism, which collects the bidders valuations in the form of bids and then
computes the market equilibrium based on those bids: 

\bigskip

\noindent \textsc{Theorem} (informal) \emph{
	Any Nash equilibrium of the Fisher market mechanism approximate the  optimal Nash social welfare within a factor of 2
	for linear valuations. For Leontief valuations, the approximation degrades linearly with the number of players. Every Nash equilibrium of the mechanism is proportional for all concave valuations.
} 
\bigskip

These bounds reveal significant differences between the quality of the Nash equilibria of the Fisher market mechanism 
for complements and substitutes.

Much more strikingly, we find that a classic mechanism known as \emph{Trading Post}, originally introduced by Shapley and Shubik
\cite{ss-tp} and studied in a long line of work in different scenarios \cite{JP,DG,MS,KS}, offers very strong guarantees. At a high
level, rather than collecting the players' preferences and computing the market equilibrium, the Trading Post mechanism gives 
each participant direct control over how to spend its budget. Once the agents choose how to distribute their budget over the
available goods, they receive a fraction from each good that is proportional to the amount they spent on it.
Our results for the Trading Post mechanism are: 

\bigskip
\noindent \textsc{Theorem} (informal) \emph{
	Any Nash equilibrium of the Trading Post mechanism approximates the optimal Nash social welfare within a factor of 2 
	for all concave valuations. For Leontief valuations, the Trading Post mechanism achieves in the equilibrium an approximation of $1 + \epsilon$ for every
	$\epsilon > 0$. Moreover, all the Nash equilibria of Trading Post are pure for all concave valuations, proportional when the valuations are concave and increasing,
	and their existence is guaranteed for all CES valuations.} 
\bigskip

In other words, not only does the Trading Post mechanism achieve the same approximation as the Fisher market mechanism 
for additive valuations, but the bound holds for all concave valuations, and is arbitrarily close to the optimum for Leontief valuations!
We view this as an important result that testifies 
to the usefulness and robustness of the Trading Post mechanism. 
A good Nash welfare approximation implies the geometric mean of values is high, and so utility is well distributed across the participants. Trading Post also ensures proportionality in any equilibrium outcome, thus the mechanism guarantees a surprising combination of both individual and aggregate fairness for very general utilities. 


An interpretation of the Leontief result is that the Trading Post mechanism limits the extent to which an agent can
affect the outcome, thus also limiting the extent to which things can go awry. Specifically, when an agent deviates in
the Trading Post mechanism, this deviation has no effect on the way that the other agents are spending their money. On
the other hand, an agent's deviation in the Fisher market mechanism can lead to a market equilibrium where the other
agents' spending and allocation has changed significantly. In addition to this, in the Fisher market mechanism an agent 
can affect the price of an item even if the agent does not end up spending on that item in the final outcome. This is in
contrast to the Trading Post mechanism where an agent can affect only the prices of the items that this agent is spending
on, so the agents are forced to ``put their money where their mouth is''.

Finally, we prove that the set of mixed Nash equilibria of the Trading Post mechanism coincides with the set of pure Nash 
equilibria \rn{even with concave valuations}, which extends our approximation bounds for this mechanism to mixed PoA. 
Moreover, both mechanisms achieve a classic notion of
individual fairness in the equilibrium, namely proportionality. 
All of the above results work for the weighted version of the Nash social welfare, where
agent $i$ has a $w_i$ fraction of the total budget, and the NSW objective becomes $\Pi_i \left[u_i(\vec{x})\right]^{w_i}$. 
\rn{In the process of obtaining a near optimal bound for Leontief in the Trading Post mechanism, we show that $\epsilon$-approximate 
	market equilibria for Leontief utilities approximate its Nash social welfare by a factor of $\frac{1}{(1+\epsilon)}$. 
	We believe that this as well as the technique that we developed to show the bounds of 2 may be of independent interest.
}

\subsection{Related Work}
The paper most closely related to our work is that of~\cite{CGG13a} which proposes \emph{truthful}
mechanisms for approximately maximizing the Nash social welfare objective. One of the truthful
mechanisms that they propose, the Partial Allocation mechanism, guarantees a $2.718$ approximation 
of the optimal NSW for both linear and Leontief valuations. In fact, the Partial Allocation 
mechanism guarantees that \emph{every agent} receives a $2.718$ approximation of the value that it 
would receive in the market equilibrium. But, in order to ensure truthfulness, this mechanism is 
forced to keep some of the goods unallocated, which makes it inapplicable for many real world settings. 
Complementing this mechanism, our work analyzes simple and well-studied mechanisms that allocate everything.

Most of the literature on fair division starting from the 1940's deals with the cake-cutting 
problem, which models the allocation of a divisible heterogeneous resource such as land, time, and
mineral deposits, among agents with different preferences~\cite{Young94,BramsT96,RobertsonWebb98,Moulin03,Barbanel04}.
Some recent work has studied the agents' incentives in cake cutting.
In particular, \cite{CLPP10} study truthful cake-cutting with agents 
having piecewise uniform valuations and provide a polynomial-time mechanism that is truthful, 
proportional, and envy-free, while \cite{MosselT10} shows that for general valuations there exists a protocol that is truthful in expectation,
envy-free, and proportional for any number of players.
The work of \cite{MN12} shows that truthfulness comes at a significant
cost in terms of efficiency for direct revelation mechanisms, while \cite{BM15} show that 
the only strategyproof mechanisms in the standard query model for cake cutting are dictatorships (even for two players; a similar
impossiblity holds for $n > 2$).  The standard cake cutting model assumes additive valuations, and so it does not capture resources
with Leontief valuations, which we also analyze in this paper.


The resource allocation literature has seen a resurgence of work studying fair and efficient
allocation for Leontief valuations~\cite{GZH11,DFH12,PPS12,GN12}. These valuations exhibit perfect
complements and they are considered to be natural valuation abstractions for computing settings where
jobs need resources in fixed ratios. \cite{GZH11} defined the notion of Dominant Resource Fairness (DRF), 
which is a generalization of the egalitarian social welfare to multiple types of resources. This 
solution has the advantage that it can be implemented truthfully for this specific class of valuations.
\cite{PPS12} assessed DRF in terms of the resulting efficiency, showing that it performs
poorly. \cite{DFH12} proposed an alternate fairness criterion called Bottleneck Based Fairness, 
which was subsequently showed by \cite{GN12} to be satisfied by the proportionally fair allocation. \cite{GN12}
also posed the study of incentives related to this latter notion as an interesting open problem.
It is worth noting that \cite{GZH11} acknowledge that the CEEI, i.e., the NSW maximizing allocation
would actually be the preferred fair division mechanism in their setting, and that the main drawback 
of this solution is the fact that it cannot be implemented truthfully. Our results show that the
Trading Post mechanism can, in fact, approximate the CEEI outcome arbitrarily well, thus shedding
new light on this setting.

The Trading Post mechanism, also known as the Shapley-Shubik game~\cite{ss-tp}, 
has been studied in an extensive body of literature over the years, sometimes under very different names, 
such as Chinese auction~\cite{Matros07}, proportional sharing mechanism (see, e.g., \cite{FLZ09}), and the 
Tullock contest in rent seeking~\cite{Tullock80,Fang02,Moldovanu01}, the latter being a variant 
of the game with a different success probability for items that nobody bid on. Trading Post can 
also be interpreted as a congestion game (see, e.g., ~\cite{GPP06}), or an all-pay auction 
when the budgets are intrinsically valuable to the players.

The fact that, facing the Fisher mechanism, the agents may gain by bidding strategically is well known.
\cite{ABGMS} studied the agents' incentives and proved existence and structural properties 
of Nash equilibria for this mechanism. Extending this work, \cite{CDZ11, CDZZ12} proved 
bounds on the extent to which an agent can gain by misreporting for various classes of valuation 
functions, including additive and Leontief. Finally, \cite{BCDFFZ14} showed bounds for the price of anarchy
of this mechanism with respect to the social welfare objective, and \cite{CT} studied large markets 
under mild randomness and showed that this price of anarchy converges to one. 

%
Finally, recent work on the NSW has revealed additional appealing properties of this objective. For indivisible
items the NSW can be approximated in polynomial time~\cite{CG15, CDGJMVY17, Anari1, Anari2}
and its optimal allocation is approximately envy-free~\cite{rCKMPSW16}. On the other hand, for divisible items, it 
can be used as an intermediate step toward approximating the normalized social welfare objective~\cite{CGG13b}.

%
%
%



\section{Preliminaries}\label{sec:prel}

Let $N = \{1, \ldots, n\}$ be a set of players (agents) and $M = \{1, \ldots, m\}$ a set of divisible goods. 
Player $i$'s utility for a bundle of goods is represented by a non-decreasing \rn{non-negative concave} valuation function $u_i:[0,1]^m
\rightarrow \Rplus$. An allocation $\vec{x}$ is a partition of the goods to the players such that $x_{i,j}$ 
represents the amount of good $j$ received by player $i$. Our goal will be to allocate all the resources fully; 
it is without loss of generality to assume that a single unit of each good is available, thus the set of
feasible allocations is $\mathcal{F}=\left\{\vec{x}~|~x_{i,j}\geq 0 \text{ and } \sum_{i=1}^{n} x_{i,j}= 1\right\}$.

Our measure for assessing the quality of an allocation is its Nash social welfare. 
At a given allocation $\vec{x}$ it is defined as follows 
\[
\mbox{NSW}(\vec{x}) = \left( \prod_{i=1}^{n} u_i(\vec{x}_i)\right)^{\frac{1}{n}}.
\]

In order to also capture situations where the 
agents may have different importance or priority, such as clout in bargaining scenarios, we also consider the weighted version of the Nash social welfare objective. Note this is the objective maximized by the Fisher market equilibrium solution when the buyers have different budgets. We slightly abuse notation and refer to the weighted objective as the Nash social welfare (NSW) as well. If $\budget_i\geq 1$ is the budget of agent $i$ and 
$\totbudget =\sum_{i=1}^n \budget_i$ is the total budget, the induced market equilibrium in Fisher's model 
maximizes the objective:
\begin{equation*}
\mbox{NSW}(\vec{x}) = \left( \prod_{i=1}^{n} u_i(\vec{x}_i)^{\budget_i}\right)^{\frac{1}{\totbudget}}.
\end{equation*}

Note that we get back the original definition when all players have the same budget. 
We would like to find mechanisms that maximize the NSW objective in the presence of strategic 
agents whose goal is to maximize their own utility.

We measure the quality of the mechanisms using the {\em price of anarchy}~\cite{NRTV07} with respect
to the NSW objective. Given a problem instance $\cal I$ and some mechanism $\cal M$ that yields 
a set of pure Nash equilibria $E$, the price of anarchy (PoA) of $\cal M$ for $\cal I$ 
is the maximum ratio between the optimal NSW---obtained at some allocation $\vec{x}^{*}$---and 
the NSW at an allocation $\vec{x} \in E$:
$
\mbox{PoA}({\cal M, I}) ~=~ \max_{\vec{x}\in E}\left\{\frac{\text{NSW}(\vec{x}^*)}{\text{NSW}(\vec{x})}\right\} 
$. The price of anarchy of $\cal M$ is the maximum of this value over all possible instances: $\max_{\cal I}\{ \text{PoA}(\cal M, \cal I)\}$.

\medskip

\noindent{\bf Valuation Functions.} We start with two very common and extensively studied valuation functions that lie at
the two extremes of a spectrum: {\em perfect substitutes} and {\em perfect complements}. For both, let $\vv_i=(v_{i,1},\dots, v_{i,m}) \in
\mathbb{R}_+^m$ be a vector of valuations for agent $i$, where $v_{i,j}$ captures the liking of agent $i$ for good $j$.  


{\em Perfect substitutes}, defined mathematically through \emph{linear (additive) valuations}, 
represent goods that can replace each other in
consumption, such as Pepsi and Coca-Cola. In the additive model, the utility of a player $i$ for bundle $\vec{x}_i$ is $u_i(\vec{x}_i)
= \sum_{j=1}^{m} v_{i,j} \cdot x_{i,j}$. 

\emph{Perfect complements}, represented by \emph{Leontief utilities}, capture scenarios where one good
may have no value without the other, such as a left and a right shoe, or the CPU time and computer memory required for the
completion of a computing task. In the Leontief model, the utility of a player $i$ for a bundle $\vec{x}_i$ is $u_i(\vec{x}_i) = \min_{j=1}^{m} \left\{x_{i,j}/v_{i,j}\right\}$; that is, player $i$ desires the items in the ratio $v_{i,1}:v_{i,2}: \ldots :v_{i,m}$. Coefficients can be rescaled freely in the Leontief model, so
w.l.o.g. $v_{i,j} \geq 1$.

Finally, perfect substitutes and complements are extreme points of a much more general class of valuations known as CES (constant elasticity of substitution), mathematically defined as 
\[
u_i(\vec{x}_i)=\big(\sum_{j=1}^m
v_{ij} \cdot x_{ij}^\rho\big)^\frac{1}{\rho}
\]
where $\rho$ parameterizes the family, and $-\infty <\rho\le 1, \;
\rho\neq 0$. The Leontief and additive utilities are obtained when $\rho$ approaches $-\infty$ and equals 1, respectively.

\section{The Fisher Market Mechanism}


In the Fisher market model, given prices $\pp=(p_1,\dots,p_m)$, each buyer $i$ demands a bundle $\xx_i$ 
that maximizes her utility subject to her budget constraints; we call this an optimal bundle of buyer $i$ 
at prices $\pp$. Prices $\pp$ induce a market equilibrium if buyers get their optimal bundle and market 
clears. Formally, prices $\pp$ and allocation $\xx$ constitute a market equilibrium if 
\begin{enumerate}
	\item Optimal bundle: $\forall i \in N$ and $\forall \yy: \yy\cdot \pp \le \budget_i$, $\; u_i(\vec{x}_i)\geq u_i(\yy)$
	\item Market clearing: Each good is fully sold or has price zero, i.e., $\forall j \in M$, $\sum_{j=1}^m x_{i,j}\le 1$, and equality
	holds if $p_j>0$. Each buyer exhausts all its budget, i.e., $\forall i \in N, \ \sum_{j=1}^m x_{i,j}p_j = \budget_i$.
\end{enumerate}

For linear and Leontief valuations, the market equilibria can be computed using
the Eisenberg-Gale (EG) convex program formulations that follow.

\begin{equation}\label{eq:eg}
\begin{array}{cc}
\mbox{\bf Linear} & \mbox{\bf Leontief}\\
\begin{array}{rcl}
\mbox{max } \ \ & & \sum^{n}_{i=1} \budget_i \cdot \log u_i\\
\mbox{s.t. } \ \ & & u_i = \sum_{j=1}^m v_{i,j} x_{i,j}, \forall i \in N\\
& & \sum^n_{i=1} x_{i,j}\leq 1,\ \ \forall \ j \in M \\
& & x_{i,j}\geq 0,\ \ \forall \ i \in N, j \in M
\end{array}
\ \ \ \ \ \ \ & \ \ \ \ \ \ \ 
\begin{array}{rcl}
\mbox{max } \ \ & & \sum^{n}_{i=1} \budget_i \cdot \log u_i \\
\mbox{s.t. } \ \ & & u_i \le \frac{x_{i,j}}{v_{i,j}},\ \forall i \in N, j \in M \\
& & \sum^n_{i=1} x_{i,j}\leq 1,\ \ \forall \ j \in M \\
& & x_{i,j}\geq 0,\ \ \forall \ i \in N, j \in M
\end{array}
\end{array}
\end{equation}

Let $p_j$ be the dual variable of the second inequality (for good $j$) in both cases, which corresponds to price of good $j$. 
Since due to strong duality Karush-Kuhn-Tucker (KKT) conditions capture solutions of the formulations, we get the following
characterization for market equilibria.

For linear valuations, an outcome $(\xx,\pp)$ is a market equilibrium if and only if, 
\begin{itemize}
	\item[$\lin_1$] $\forall i \in N$ and $\forall j \in M$, $x_{i,j}>0 \Rightarrow \frac{v_{i,j}}{p_j} = \max_{k \in M} \frac{v_{i,k}}{p_k}$.
	\item[$\lin_2$] $\forall i \in N$, $\sum_{j\in M} x_{i,j} p_j=\budget_i$. $\forall j \in M$, either $\sum_{i \in N} x_{i,j} = 1$ or
	$p_j=0$.
\end{itemize}
\medskip

For Leontief valuations, an outcome $(\xx,\pp)$ is a market equilibrium if and only if.
\begin{itemize}
	\item[$\leon_1$] $\forall i \in N$ and $\forall j \in M$, $v_{i,j}>0 \Rightarrow u_i=\frac{x_{i,j}}{v_{i,j}}=\frac{\budget_i}{\sum_{j \in M}
		v_{i,j}p_j}$. Note that $\sum_{j \in M} v_{i,j}p_j$ is the amount buyer $i$ has to spend to get unit utility.
	\item[$\leon_2$] $\forall i \in N$, $\sum_{j\in M} x_{i,j} p_j=\budget_i$. $\forall j \in M$, either $\sum_{i \in N} x_{i,j} = 1$ or
	$p_j=0$. 
\end{itemize}

The Fisher market mechanism asks that the agents report their valuations and then
it computes the market equilibrium allocation with respect to the reported valuations
using the EG formulations.

\begin{definition}[Fisher Market Mechanism]
	The \emph{Fisher Market Mechanism} is such that:
	\begin{itemize}
		\item The strategy space of each agent $i$ consists of all possible valuations the agent may pose: $S_i = \{\vec{s}_i \; | \; \vec{s}_i \in \mathbb{R}^{m}_{\ge 0}\}$. We refer to an agent's strategy as a \emph{report}.
		\item Given a strategy profile $\vec{s} = (\vec{s}_i)_{i=1}^{n}$, the outcome of the game is a market equilibrium of the Fisher market given by $\langle \budget_i, \vec{s}_i \rangle$, after removing the items $j$ for which $\sum_{i \in N} s_i(j)=0$. 
		If there exists a market equilibrium $\mathcal{E}$ preferred by all the agents (with respect to their true valuations), then $\mathcal{E}$ is the outcome of the game on $\langle \budget_i, \vec{s}_i\rangle$. Otherwise,
		the outcome is any fixed market equilibrium. 
	\end{itemize}
\end{definition}

It is well known that the buyers have incentives to hide their true valuations in a Fisher market mechanism. We illustrate this
phenomenon through an example.

\begin{example}
	Consider a Fisher market with players $N = \{1, 2\}$, items $M= \{1,2\}$,
	additive valuations $v_{1,1} = 1$, $v_{1,2} = 0$, $v_{2,1} = v_{2,2} = 0.5$, and budgets equal to $1$.
	If the players are truthful, the market equilibrium allocation is $\vec{x}_1 = (1, 0)$, $\vec{x}_2 = (0,1)$.
	However, if player $2$ pretended that its value for item $2$ is a very small $v_{2,2}' = \epsilon > 0$,
	then player $2$ would not only get item $2$, but also a fraction of item $1$.
\end{example}

In the rest of this section we study the performance of the Fisher mechanism when the goods are substitutes and complements, respectively.

\subsection{Fisher Market: Perfect Substitutes}
In this section we study efficiency loss in the Fisher market due to strategic agents with additive valuations. 
We note that pure Nash equilibria in the induced game are known to always exist due to \cite{ABGMS}.\footnote{The 
	existence is shown under a {\em conflict-free} tie breaking rule, which tries to allocate best bundle to as many 
	agents as possible when there is a choice.}
Our first main result states that the Fisher market approximates the Nash Social Welfare 
within a small constant factor, even when the players are strategic.

\begin{theorem} \label{thm:Fisher_additive_upper_bound}
	The Fisher Market Mechanism with linear valuations has price of anarchy at most 2.
\end{theorem}

We first prove a useful lemma, which bounds the change in prices due to a unilateral deviation in the Fisher market with additive valuations.

\begin{lemma}\label{lem:price_change}
	Let $\vec{p}$ be the prices in a Fisher market equilibrium. Suppose that buyer $i$ unilaterally changes its reported values to
	$\vec{v}_i'$, leading to new market equilibrium prices $\vec{p}'$. Then, 
	\[\sum_{j:\; p'_j >p_j}p'_j ~\leq~ \budget_i+\sum_{j:\; p'_j >p_j}p_j.\]
\end{lemma}
\begin{proof}
	
	Let $M^+$, $M^-$, and $M^=$ be the sets of goods whose prices have strictly increased, strictly decreased, and 
	remained unchanged, respectively, when transitioning from $\vec{p}$ to $\vec{p}'$. 
	If some player $k\neq i$ is buying any fraction of a good from $M^-$ or $M^=$ at $\vec{p}'$, then the player 
	cannot be buying anything from $M^+$ at $\vec{p}'$. To verify this fact, assume that player $k$ was 
	spending on some item $\alpha\in M^- \cup M^=$ at prices $\vec{p}$ and is now spending on some item $\beta\in M^+$
	at prices $\vec{p}'$. Since both these allocations are market equilibria, $\lin_1$ implies that 
	$\frac{v_{k,\alpha}}{p_{\alpha}} \geq \frac{v_{k,\beta}}{p_{\beta}}$ at prices $\vec{p}$ and
	$\frac{v_{k,\alpha}}{p'_{\alpha}} \leq \frac{v_{k,\beta}}{p'_{\beta}}$ at prices $\vec{p}'$. Since $p'_\alpha \leq p_\alpha$
	and $p'_\beta > p_\beta$, this leads to the contradiction that $v_{k,\alpha}> v_{k,\alpha}$.
	
	Thus, any player, other than $i$, who is buying goods from $M^+$ at prices $\vec{p}'$ had to be spending all
	its budget on these items at prices $\vec{p}$. This implies that the only reason why the sum of the prices
	of the goods in $M^+$ could increase is because player $i$ is contributing more money on these items. Since the budget of $i$ is $B_i$, the total increase in these prices is at most $B_i$. This concludes the proof.
\end{proof}

In addition to this, we will be using a folklore weighted arithmetic and geometric mean inequality stated in the following lemma.
\begin{lemma}\label{lem:budgets_bound}
	For any nonnegative numbers $\rho_1, \rho_2,\dots, \rho_n$ and $w_1, w_2, \dots, w_n$ with
	$W = \sum_{i=1}^n w_i$, we have 
	$\left(\prod_{i=1}^n \rho_i^{w_i}\right)^{1/W} ~\leq~ \frac{\sum_{i=1}^n \rho_i w_i}{W}.$
\end{lemma}


We can now prove the main theorem.

\begin{proof}(of \thmref{thm:Fisher_additive_upper_bound})
	Given a problem instance with additive valuations, 
	let $\vec{x}^*$ be the allocation that maximizes the Nash social welfare---i.e., the market equilibrium allocation
	with respect to the true valuations---and $\tilde{\vec{x}}$ and $\tilde{\vec{p}}$ the allocation 
	and prices respectively obtained under some Nash equilibrium of the market, where the players report fake valuations
	$\tilde{\vec{v}}$. Additionally, for each player $i$, let $\vec{x}^i$ be the allocation that would 
	arise if every player $k\neq i$ reported $\tilde{\vec{v}}_k$ while $i$ reported its true 
	value $v_{ij}$ for every item $j$ with $x^*_{ij}>0$, and zero elsewhere. 
	
	Since $\tilde{\vec{v}}$ is an equilibrium, this unilateral deviation of $i$ cannot
	increase its utility:
	\begin{equation}\label{ineq:NashEq}
	u_i(\tilde{\vec{x}}_i) ~\geq~ u_i(\vec{x}^i_{i}) ~=~ \sum_{j=1}^{m} x^i_{ij} v_{ij}.
	\end{equation}
	
	Since $\vec{x}^i$ is a market equilibrium with respect to the reported values,
	according to the KKT condition ($\lin_1$), $x^i_{ij}>0$ implies 
	$v_{ij}/p^i_j \geq v_{ik}/p^i_k$ for any other item $k$, where $p^i_j$ is 
	the price of item $j$ in $\vec{x}^i$. Also, the KKT condition ($\lin_2$) imply that 
	$\sum_j x^i_{ij}p^i_j =\budget_i$ for every bidder $i$. 
	Therefore, $\vec{x}^i_i$ provides $i$ at least as much value as any other bundle 
	that $i$ can afford facing prices $p^i$, i.e., any bundle that costs at most $B_i$. 
	In particular, consider the allocation $\vec{x}'$ such that
	\[
	x'_{ij} = \frac{\budget_i x^*_{ij}}{\sum_{k=1}^{m} x^*_{ik}p^i_k}.
	\]
	For this allocation we have:
	\[
	\sum_{j=1}^{m} x'_{ij}p^i_j ~=~ \sum_{j=1}^{m} \left( \frac{\budget_i x^*_{ij}}{\sum_{k=1}^{m} x^*_{ik}p^i_k} \right) \cdot p^i_j
	~=~ \frac{\budget_i \sum_{j=1}^m x^*_{ij}p^i_j}{\sum_{k=1}^{m} x^*_{ik}p^i_k} ~=~\budget_i.
	\]
	Therefore, player $i$ can afford this bundle of items using the budget $\budget_i$,
	which implies that $\vec{x}^i_i$ provides $i$ at least as much value, i.e.,
	\[u_i(\vec{x}^i_{i}) ~\geq~ u_i(\vec{x}'_i) ~=~ \sum_{j=1}^{m} v_{i,j}x'_{i,j} = \frac{\budget_i u_i(\vec{x}^*_{i})}{\sum_{j=1}^{m}
		x^*_{i,j}p^i_j}.\]
	Let $\rho_i$ denote the ratio $\frac{u_i(\vec{x}^*)}{u_i(\tilde{\vec{x}}_i)}$. 
	Using the Nash equilibrium inequality~\eqref{ineq:NashEq}, we get:
	\begin{equation*}
	u_i(\tilde{\vec{x}}_i) \geq \frac{\budget_i u_i(\vec{x}^*_i)}{\sum_{j=1}^{m} x^*_{i,j}p^i_j} ~\Rightarrow~ 
	\frac{u_i(\vec{x}^*)}{u_i(\tilde{\vec{x}}_i)}\leq \frac{\sum_{j=1}^{m} x^*_{i,j}p^i_j}{\budget_i} ~\Rightarrow~ \rho_i \budget_i \le
	\sum_{j=1}^m x^*_{i,j}p^i_j.
	\end{equation*}
	Using \lemref{lem:price_change}, we get:
	\begin{equation*}
	\rho_i\budget_i ~\leq~ 
	\sum_{j=1}^{m} x^*_{i,j}p^i_j ~\leq~ 
	\sum_{j:\; p^i_j \leq \eqp_j} x^*_{i,j}\eqp_j + \sum_{j:\; p^i_j > \eqp_j} x^*_{i,j}p^i_j ~\leq~ 
	\budget_i +\sum_{j=1}^{m} x^*_{i,j}\eqp_j, 
	\end{equation*}
	and summing over all players:
	\begin{equation}\label{ineq:RatioBound}
	\sum_{i=1}^{n}\rho_i\budget_i ~\leq~ 
	\sum_{i=1}^{n}\left(\budget_i +\sum_{j=1}^{m} x^*_{i,j}\eqp_j\right) ~\leq~  
	2\totbudget.
	\end{equation}

	
	Substituting $\budget_i$ for $w_i$ in~\lemref{lem:budgets_bound} and using Inequality~\eqref{ineq:RatioBound} yields:
	\[
	\left(\prod_{i=1}^n \rho_i^{\budget_i}\right)^{1/\totbudget} ~\le~ \frac{\sum_{i=1}^n \rho_i \budget_i}{\totbudget} ~\le~ 2.
	\] 
\end{proof}

Next we show a lower bound for the price of anarchy of the Fisher mechanism. 
We construct a collection of problem instances whose PoA goes to $e^{1/e}$
as number of players grows.

\begin{theorem}\label{thm:FisherLB}
	The Fisher Market mechanism with additive valuations has a price of anarchy no better than $e^{1/e}\approx 1.445$.
\end{theorem}
Given some value of $n$, we construct a market with $n+2$ agents and $n+1$ goods. 
Fix an integer $k\le n$; we will set its value later. Each player $i\le k$ likes 
only good $i$, {\em i.e.,} $v_{i,i}=1$ and all other $v_{i,j}=0$. On the other hand, 
every agent $i\in [k+1, n]$, apart from having $v_{i,i}=1$, also has some small, but 
positive, value $v_{i,j}=\epsilon$ for all items $j\leq k$. The rest of that agent's 
$v_{i,j}$ values are zero. Agent $n+1$ has a small but positive value $\epsilon'$ for 
goods $j \in [k+1, n]$ and value 2 for good $n+1$. Finally, agent $n+2$ values only 
good $n+1$ at value 2. Here $\epsilon < < \epsilon'$, and we will set their values later. 
In the allocation where every agent $i \in [1, n]$ gets all of good $i$, while agents 
$n+1$ and $n+2$ share good $n+1$ equally, the NSW is equal to 1. Next, we construct a 
Nash equilibrium strategy profile $\sss$ of the above market where the NSW approaches 
$(1/e)^{1/e}$ as $n \rightarrow \infty$. 

We define a strategy profile $\sss$ where the first $k$ agents and the last agent,
i.e., agent $n+2$, bid truthfully, while every bidder $i\in[k+1,n]$ misreports in 
a way such that it ends up spending some small amount $\delta=2\epsilon'$ on item 
$i$ and the rest of its budget, namely $(1-\delta)$, is equally divided on items in $[1,k]$. 
We later set a value for $\delta$ such that agent $n+1$ would want to buy only good $n+1$. 

We now show that the above profile is a Nash equilibrium for carefully chosen values of 
$\epsilon$, and $\epsilon'$. Note that, since bidders $i\in [1,k]$ and bidder $n+2$ bids 
truthfully and values just one item, each one of these bidders will spend all of its budget 
on the corresponding items, no matter what the remaining bidders report. Therefore, the 
price of items $j\in [1,k]$ and item $n+1$ will be exactly 1 if we exclude the spending 
of bidders $i\in [k+1,n+1]$, no matter what these bidders report. Also, if the price of 
items $j\in [k+1, n]$ is equal to $\delta$ in profile $\sss$, then this price will not
drop below $\delta$, irrespective of what bidder $n+1$ reports. In the next two lemmas,  
we show that a deviation from $\sss$ does not help the bidders, even when the prices are 
held constant for the goods where an agent starts spending more money. 

\begin{lemma}\label{lem:flb1}
	If $\delta= 2\epsilon'$ then, even if the prices of goods $j\in [k+1, n]$ are fixed at $\delta$, the 
	$(n+1)^{th}$ bidder has no incentive to spend money on any good other than good $n+1$.  
\end{lemma}
\begin{proof}
	If agent $n+1$ spends all of its budget on item $n+1$, then its price becomes 2
	and the agent receives half of that item, i.e., a utility of 1. On the other hand, if the agent
	spends at total of $\gamma>0$ on goods $j \in [k+1,n]$ then, even if the price of these items 
	remains $\delta$, her utility would be $\epsilon'\frac{\gamma}{\delta} +2 \frac{1-\gamma}{2-\gamma}$. 
	For this to be strictly greater than 1 we have to have $\delta<2\epsilon'-\gamma\epsilon'$, which
	contradicts our assumption that $\delta = 2\epsilon'$.
\end{proof}

Furthermore, using a similar analysis as that of Lemma \ref{lem:flb1}, it follows that if agent $i\in[k+1, n]$
deviates in a way that decreases its spending on good $i$ to $2\epsilon'-\tau$ for $\tau>0$, then the Fisher 
market mechanism outcome will have the $(n+1)^{th}$ agent spending at least $\frac{\tau}{1+\epsilon'}$ on that 
item. Thus, such a deviation would cause agent $i$ to lose its monopoly on good $i$. Next, we show that this is 
not advantageous for agent $i$ if we set $\epsilon=\frac{1}{n^4}$ and $\epsilon'=\frac{1}{n}$, which implies
$\delta=\frac{2}{n}$. 

\begin{lemma}\label{lem:flb2}
	For any $\tau\ge 0$, if agent $i$ is spending $\delta -\tau$ on good $i$ and agent $(n+1)$ is spending 
	$\frac{\tau}{1+\epsilon'}$, when others are bidding according to $\sss$, then agent $i$'s utility is maximized at $\tau=0$. 
\end{lemma}
\begin{proof}
	Note that at $\tau=0$, agent $i\in[k+1, n]$ is spending $\delta=2\epsilon'$ on good $i$ and, according to 
	\lemref{lem:flb1}, agent $n+1$ would not be interested in spending on good $i$. Thus, agent $i$ is buying 
	good $i$ exclusively and spends $\frac{1-\delta}{k}$ on each of the first $k$ goods. Thus the price each one
	of the first $k$ goods is $1+(n-k)\frac{1-\delta}{k}$. Agent $i$'s utility at this allocation is 
	\[1+ \frac{\frac{1-\delta}{k} }{1+(n-k)\frac{1-\delta}{k}} k\epsilon  ~=~ 
	1+\frac{(1-\delta)k\epsilon}{k+(n-k)(1-\delta)} \]
	
	As mentioned above, if agent $i$ reduces spending on good $i$ by $\tau$, then the $(n+1)^{th}$ agent will 
	end up spending at least $\frac{\tau}{1+\epsilon'}$ on this item. This may lead to increased prices for the
	first $k$ goods, but we show that this deviation would not benefit $i$ even if these prices remained the same. 
	The utility of agent $i$ after such a deviation would be
	\[\frac{\delta-\tau}{\delta-\tau + \frac{\tau}{1+\epsilon'}} + 
	\frac{(1-\delta+\tau) k\epsilon}{k+(n-k)(1-\delta)}\]
	%
	The latter utility is greater than the former only when 
	\[\frac{\tau k\epsilon}{k+(n-k)(1-\delta)} ~>~ \frac{\frac{\tau}{1+\epsilon'}}{\delta-\tau + \frac{\tau}{1+\epsilon'}} ~~~~\Rightarrow~~~~
	\epsilon ~>~ \frac{k+(n-k)(1-\delta)}{k(\delta+\delta\epsilon'-\tau\epsilon')} ~>~1,\]
	which contradicts the fact that $\epsilon=\frac{1}{n^4}$. 
\end{proof}

Lemmas \ref{lem:flb1} and \ref{lem:flb2} imply that $\sss$ is a Nash equilibrium. The price of
the first $k$ goods at this Nash equilibrium is $1+\frac{(n-k)(1-\delta)}{k} =
\frac{k+(n-k)(1-\delta)}{k}$. Thus, the utility of buyer $i \in[1, k]$, who spends all of its 
\$1 on good $i$, is $u_i=\frac{k}{k+(n-k)(1-\delta)}$. On the other hand, the utility of buyer 
$i\in[k+1, n]$, who gets all of good $i$ and some of the first $k$ goods, is 
$u_i=1+ \frac{(1-\delta)k\epsilon}{k+(n-k)(1-\delta)}$.  Agents $(n+1)$ and $(n+2)$ get half of good $n+1$ and 
thereby get utility of 1 each. Since $\epsilon=\frac{1}{n^4}$ and $\epsilon'=\frac{1}{n}$, then 
$\forall i \in[1, k],\ \lim_{n\rightarrow \infty} u_i = \frac{k}{n}$ and 
$\forall i \in[k+1, n] \lim_{n\rightarrow \infty} u_i = 1$.  Thus, NSW at this bid profile as 
$n\rightarrow \infty$ is $(\frac{k}{n})^{\frac{k}{n}}$, and thereby PoA is at least $(\frac{n}{k})^{k/n}$. 
Letting $k = \frac{n}{e}$, this becomes $e^{1/e}$, which concludes the proof of Theorem~\ref{thm:FisherLB}.

\subsection{Fisher Market: Perfect Complements}

On the other hand, the Fisher market with perfect complements has a price of anarchy that grows linearly with the number of players. 


\begin{theorem} \label{thm:Fisher_Leontief_poa}
	The Fisher Market Mechanism with Leontief valuations has a price of anarchy of $n$, the number of players, and the bound is tight.
\end{theorem}

{\textsc{Remark.}
	\rn{We observe that PoA of Fisher Market Mechanism degrades continuously within CES as we move from additive to Leontief. For the complements range, {\em i.e.,} $\rho <0$, the example achieving $n$ lower bound for Leontief can be modified to show that as $\rho$ decreases the PoA increases. The modification is in the Nash equilibrium profile, where agent $i$ reports $v_{i,i}=1+\epsilon$ instead of $1$ for an appropriately chosen $\epsilon\ge 0$ depending on $\rho$ and number of players $n$. }

\section{The Trading Post Mechanism}


The important difference between the Trading Post mechanism and the Fisher Market mechanism is the strategy space of the agents.
More precisely, unlike the Fisher market mechanism, where the agents' are asked to report their valuations, the Trading Post 
mechanism instead asks the agents to directly choose how to distribute their budgets. Once the agents have chosen how much of their
budget to spend on each of the goods, the total spending on each good $j$ is treated as its price, and each agent $i$ is allocated a
fraction of good $j$ proportional to the amount that $i$ is spending on $j$. Therefore, the strategy set of each player $i$ is 
$\mathcal{S}_i = \left\{\vec{b}_i \in [0,B_i]^{m}\; | \; \sum_{j=1}^{m} b_{i,j} = \budget_i\right\}.$ 
%

Given a bid profile $\vec{b} = (\vec{b}_1, \ldots, \vec{b}_n)$, the induced allocation is:
\begin{equation*}
x_{i,j} = \left\{
\begin{array}{ll}
\frac{b_{i,j}}{\sum_{k=1}^{n} b_{k,j}} & \mbox{if} \; b_{i,j} > 0\\
0 & \mbox{otherwise}
\end{array}
\right.
\end{equation*}

\rn{In this section we analyze Trading Post for additive and Leontief valuations, extending a number of results to CES and arbitrary concave functions under a mild technical condition known as \emph{perfect competition}, which states that for each good $j$, there exist at least two buyers that demand strictly positive amount of $j$ when the good is priced at zero. For additive and Leontief valuations, the condition is equivalent to 
	$v_{i,j}> 0, v_{k,j} > 0$, for some $i, k \in N$, $i \neq k$. 
	If this were not satisfied, we could either discard the good or give it away for free.}


\rn{Our main results in this 
	section are that the {\em pure} Nash equilibria of Trading Post approximate the NSW objective within a factor of $2$ 
	for linear valuations, and within a factor of $1 + \epsilon$ for every $\epsilon > 0$ for Leontief valuations. In fact we extend the approximation factor of $2$ to CES and more generally to concave utilities. Moreover, in all of these cases, all the Nash equilibria of the game are pure and their existence is guaranteed for CES utilities.
}

\subsection{Trading Post: Perfect Substitutes}

The existence of pure Nash equilibria in the Trading Post mechanism for additive valuations was established by Feldman, Lai, and Zhang~\cite{FLZ09} under perfect competition. Without competition, even very simple games may not have pure Nash equilibria. To see this, consider for instance a game with two players, two items, and additive valuations $v_{1,1} = 1$, $v_{1,2} = 0$, $v_{2,1} = v_{2,2} = 0.5$. Through a case analysis it can be seen that both players will compete for item $1$, while player $2$ is the only one that wants item $2$, reason for which this player will successively reduce its bid for $2$ to get a higher fraction from item $1$. However, in the limit of its bid for the second item going to zero, player $2$ loses the item.

Our first result shows that Trading Post matches the Fisher market for additive valuations. 

\begin{theorem}\label{thm:ub_tp_subs}
	The Trading Post Mechanism with linear valuations has price of anarchy at most 2.
\end{theorem}
\begin{proof}
	Given a problem instance with linear valuations, let $\optx$ be the allocation that maximizes 
	the NSW and $p^*$ be the corresponding market equilibrium prices. Also, let $\eqx$ be 
	the allocation in a Nash equilibrium where each player $i$ bids $\eqbb_i$ and the price 
	of each item $j$ is $\eqp_j=\sum_i \eqbid_{i,j}$. 
	
	For some player $i$, let $\devx$ be the allocation that arises if every player $k\neq i$ 
	bids $\eqbb_k$ while agent $i$ unilaterally deviates to $\bid'_{i,j}$ for each item $j$. 
	In this deviation, player $i$ first withdraws all of its money, leaving the price of item 
	$j$ to be $p'_j = \eqp_j -\eqbid_{i,j}$, and then it redistributes it by spending $\bid'_{i,j}$ 
	on each item $j$. The perfect competition condition ensures that, once player $i$ withdraws its money,
	the prices remain positive ($p'_j>0,\ \forall j$), as at least two agents are interested in each good. 
	Let the new bid $\bid'_{i,j}$ be such that for some $\beta_i>0$ and every item $j$:
	\begin{equation}\label{eq:delta}
	\frac{\bid'_{i,j}}{p'_j+\bid'_{i,j}}~=~\frac{x^*_{i,j}}{\beta_i}.
	\end{equation}
	Bid $\vec{b}'_i$ is implied by the solution of the following program. 
	\[
	\text{min: } \beta_i\ \ \ \ \ \text{s.t.:}\ \ \ \  \frac{1}{\beta_i} = \frac{\bid'_{i,j}}{x^*_{i,j}(p'_j+\bid'_{i,j})}\ \  \mbox{ and }\ \  \bid'_{i,j}\ge 0 \ \ \ \forall j\in M;\ \ \ \ \ \ \ \sum_{j\in M} \bid'_{i,j} \le \budget_i
	\]

	Setting $\bb'_{i,j}=0$ and $\beta_i=\infty$ gives a feasible point in the above program, and therefore it has a minimum. 
	The allocation induced by this unilateral deviation of $i$ is 
	$x'_{i,j}~=~\frac{\bid'_{i,j}}{p'_j+\bid'_{i,j}}~=~\frac{x^*_{i,j}}{\beta_i}.$
	Therefore, the utility of player $i$ after this deviation is $u_i(\optx)/\beta_i$. 
	But the outcome $\eqx$ is a Nash equilibrium, so this deviation cannot yield a higher
	utility for $i$, which implies that $u_i(\eqx)\geq u_i(\optx)/\beta_i$.
	By definition of $\bid'$, we get $\sum_{j} \bid'_{i,j}=\sum_j x'_{i,j} (p'_j + \bid'_{i,j})= \budget_i$; all the money is spent at optimum $\beta_i$. 
	Therefore, replacing for $x'_{i,j} =x^*_{i,j}/\beta_i$, we get 
	\begin{equation}\label{eq:beta}
	\budget_i \beta_i = \sum_{j=1}^m x^*_{i,j} (p'_j + \bid'_{i,j})
	\end{equation}
	Since $0\le x^*_{i,j}\leq 1$ and $p'_j \le \eqp_j$ for each item $j$, for agent $i$ we have, 
	\begin{equation}\label{ineq:price_change}
	\sum_{j=1}^m x^*_{i,j}(p'_j+\bid'_{i,j}) \le \sum_{j=1}^m x^*_{i,j} \bid'_{i,j} + \sum_{j=1}^m x^*_{i,j} \eqp_j ~\leq~ \budget_i + \sum_{j=1}^m x^*_{i,j}\eqp_j. 
	\end{equation}
	
	Using Inequalities~\eqref{eq:beta} and~\eqref{ineq:price_change}, and summing over all 
	players gives 
	\begin{equation}\label{eq:ratio}
	\sum_{i=1}^n \budget_i\frac{u_i(\vec{x}^*_i)}{u_i(\tilde{\vec{x}}_i)} ~\leq~ 
	\sum_{i=1}^n \budget_i\beta_i ~\leq~ 
	\sum_{i=1}^n \left(\budget_i + \sum_{j=1}^m x^*_{i,j}\eqp_j\right) ~\leq~ 
	2\totbudget.
	\end{equation}
	
	\rn{Using the inequality of weighted arithmetic and geometric means provided by Lemma~\ref{lem:budgets_bound} on $\beta_i$'s of \eqref{eq:ratio}, we conclude that $\prod_{i=1}^{n}\left( \frac{u_i(\vec{x}^*_i)}{u_i(\tilde{\vec{x}}_i)}\right)^{\budget_i/\cal B} ~=~ \prod_{i=1}^{n}\left( \beta_i \right)^{\budget_i/\cal B}
		~\leq~ 2$. 
	}
\end{proof}

%

The next theorem complements the upper bound on the price of anarchy with a lower bound of approximately 1.445.

\begin{theorem} \label{thm:lower_bound_tp_subst} 
	The Trading Post mechanism has a price of anarchy no better than $e^{1/e}\approx 1.445$.
\end{theorem}
\begin{proof}
	In order to prove this theorem we can use the same family of problem instances used in the proof
	of Theorem~\ref{thm:FisherLB}. In fact, verifying that the market equilibrium outcome induced by
	the strategy profile $\sss$ in that construction is a Nash equilibrium for the Trading Post mechanism 
	as well is much more straightforward since an agent's deviation may affect only the way that particular 
	agent ends up spending its budget. 
\end{proof}

\rn{The upper and lower bounds of Trading Post with additive valuations given by Theorems \ref{thm:ub_tp_subs} and \ref{thm:lower_bound_tp_subst} are same as that of Fisher market mechanism. However, next we show this symmetry breaks as soon as we move away from additive valuations, and Trading Post fares significantly better.} 

\subsection{Trading Post: Perfect Complements}

We first characterize 
the precise conditions under which the Trading Post mechanism has exact 
pure Nash equilibria for Leontief utilities; \rn{the missing proofs of this section may be found in the full version \cite{fullversion}.}

\begin{theorem} \label{thm:TP_char_exact}
	The Trading Post mechanism with Leontief utilities has pure Nash equilibria if and only if the corresponding market 
	equilibrium prices are all strictly positive. When this happens, the Nash equilibrium utilities in Trading Post
	are unique and the price of anarchy is $1$.
\end{theorem}

This theorem shows a correspondence between the pure Nash equilibria of Trading Post and the corresponding market 
equilibria with respect to the agents' (true) valuations.
We observe however that existence of pure Nash equilibria in Trading Post is not guaranteed for Leontief utilities.
\begin{example}
	Consider a game with two players and two items, where player 1 has values $v_{1,1} = v_{1,2} = 0.5$ and player 2 has $v_{2,1} = 0.9$, $v_{2,2} = 0.1$. 
	Assume there is a pure Nash equilibrium profile $\vec{b}$. Since both players require a non-zero amount from every item for
	their utility to be positive, we have that $b_{i,j} > 0$ for all $i,j \in \{1, 2\}$. Denote $b_1 = b_{1,1}$ and 
	$b_2 = b_{2,1}$; then $b_{1,2} = 1 - b_1$ and $b_{2,1} = 1 - b_2$. Note that each player must receive the two items in the same ratio relative to its valuation; that is:
	\begin{small}
		\begin{equation}\label{eq:example_equil}
		u_i(\vec{b}) = \left( \frac{b_i}{b_1 + b_2} \right)\frac{1}{v_{i,1}}  = \left( \frac{1 - b_i}{b_1 + b_2} \right)\frac{1}{v_{i,2}}  
		\end{equation}
	\end{small}
	Otherwise, if the two ratios were not equal, then a player could transfer weight among the items to improve the smaller fraction. Then the requirement in ~\ref{eq:example_equil} are equivalent to the following equations:
	\begin{small}
		\begin{equation} \label{eq:example1a}
		\left( \frac{b_1}{b_1 + b_2} \right)\frac{1}{0.5}  = \left( \frac{1 - b_1}{b_1 + b_2} \right)\frac{1}{0.5} \iff b_1 = b_2
		\end{equation}
	\end{small}
	and 
	\begin{small}
		\begin{equation} \label{eq:example1b}
		\left( \frac{b_2}{b_1 + b_2} \right)\frac{1}{0.9}  = \left( \frac{1 - b_2}{b_1 + b_2} \right)\frac{1}{0.1} \iff 8 b_2^2 + 8 b_1 b_2 = 9 b_1 + 7 b_2
		\end{equation}
	\end{small}
	Combining equations \ref{eq:example1a} and \ref{eq:example1b}, we get that $b_1 = 1$ and $b_2 = 1$, which contradicts the requirement that $b_1,b_2 \in (0,1)$. Thus the equilibrium profile $\vec{b}$ cannot exist.
\end{example}

The issue illustrated by this example is that the Trading Post cannot implement market outcomes when there exist
items priced at zero in the corresponding market equilibrium. 
%
This motivates us to introduce an entrance fee in the Trading Post mechanism, denoted by a parameter $\Delta>0$, which is the 
minimum amount that an agent needs to spend on an item in order to receive any of it. We denote 
the corresponding mechanism by $\mathcal{TP}(\Delta)$. The value of $\Delta$ can be arbitrarily small,
so its impact on the outcome of the game is insignificant.

Formally, given a bid profile $\vec{b} = (\vec{b}_1, \ldots, \vec{b}_n)$, let $\vec{\overline{b}} = (\vec{\overline{b}}_1, \ldots, \vec{\overline{b}}_n)$ be
the ``effective'' bid profile which, for every $i\in N$ and $j\in M$, satisfies $\overline{b}_{i,j}=b_{i,j}$ if $b_{i,j}\geq \Delta$, and
$\overline{b}_{i,j}=0$ otherwise. Then, the bid profile $\vec{b}$ yields the allocation:
\begin{equation*}
x_{i,j} = \frac{\overline{b}_{i,j}}{\sum_{k=1}^{n} \overline{b}_{k,j}} \ \ \mbox{if} \; \overline{b}_{i,j} > 0,\ \ \ \mbox{ and } \ \ \ x_{i,j} =0\ \ \mbox{otherwise}
\end{equation*}
Clearly, in any Nash equilibrium, we have that for every player $i$ and item $j$ we have that  $b_{i,j}\geq \Delta$ or $b_{i,j}=0$ (the latter identity holds for those goods $j$ that are outside of player $i$'s demand). 

The main result of this section is that, for every $\epsilon\in (0,1/m)$, the Trading Post 
mechanism with $\Delta\leq \epsilon/m^2$ has a price of anarchy of at most $1+\epsilon$.
We first show that Trading Post has a pure Nash equilibrium for Leontief utilities for every strictly positive entrance fee. The proof uses an application of Glicksberg's theorem for continuous games. 

\begin{theorem} \label{thm:TP_pne_existence}
	The parameterized Trading Post mechanism $\mathcal{TP}(\Delta)$ is guaranteed to have a pure Nash equilibrium for every $\Delta > 0$.
\end{theorem}

We start by defining a notion of approximate market equilibrium that will be useful 

\begin{definition}[$\epsilon$-market equilibrium]
	Given a problem instance and some $\epsilon > 0$, an outcome $(\vec{p}, \vec{x})$ is an $\epsilon$-market equilibrium if and only if:
	$(i)$ All the goods with a positive price are completely sold. $(ii)$  All the
	buyers exhaust their budget. $(iii)$ Each buyer gets an $\epsilon$-optimal
	bundle at prices $\vec{p}$; that is, for every bundle $\vec{y} \in [0,1]^m$
	that $i$ could afford at these prices $(\vec{p}\cdot \vec{y} \leq B_i)$, we
	have $u_i(\vec{y}) \leq u_i(\vec{x}_i) (1+ \epsilon)$.
\end{definition}

The following theorem states that for every small enough $\epsilon > 0$, all the PNE of the Trading 
Post game with a small enough entrance fee correspond to $\epsilon$-market 
equilibrium outcomes. 

\begin{theorem} \label{thm:TP_apx_market}
	Let $\epsilon> 0$. Then for every $0 < \Delta < \min\left\{\frac{\epsilon}{m^2}, \frac{1}{m}\right\}$, every pure Nash equilibrium 
	of the mechanism $\mathcal{TP}(\Delta)$ with Leontief valuations corresponds to an $\epsilon$-market equilibrium.
\end{theorem}
\begin{proof}
	
	Let $\eqbb$ be a pure Nash equilibrium of $\mathcal{TP}(\Delta)$ and $\vec{x}$ the induced allocation. 
	For each player $i$, let $D_i = \{j \in M \; | \; v_{i,j} > 0\}$ be the set of items that $i$ requires, and let $m_i = |D_i|$.
	We also override notation and refer to $u_i(\vec{b})$ as the utility of player $i$ when the strategy profiles are $\vec{b}$.
	
	
	First note that $\eqbb_{i,j} > 0$ for each player $i$ and item $j \in D_i$. If this were not the case, then player $i$ would get zero utility at strategy profile $\eqbb$; this is worse than playing the uniform strategy $\vec{z}_i = (B_i/m, \ldots, B_i/m)$, 
	which guarantees $i$ a positive value regardless of the strategies of the other players $\eqbb_{-i}$, namely:
	\begin{equation*}
	u_i(\vec{z}_i, \eqbb_{-i})  ~=~ \min_{j \in D_i} \left \{ \frac{z_{i,j}}{z_{i,j} + \sum_{k \neq i} \tilde{b}_{k,j}}  \cdot \frac{1}{v_{i,j}} \right \} 
	~\geq~  \min_{j \in D_i} \left \{ \frac{B_i/m}{B_i/m + \sum_{k \neq i} B_k} \cdot \frac{1}{v_{i,j}} \right\} 
	~>~ 0
	\end{equation*}
	
	For each player $i$ and item $j \in D_i$, denote the fraction of utility that $i$ derives from $j$ by:
	$$\phi_{i,j} = \frac{\tilde{b}_{i,j}}{\sum_{k=1}^{n} \tilde{b}_{k,j}}\cdot \frac{1}{v_{i,j}}.$$
	
	Then $u_i(\eqbb) = \min_{j \in D_i} \phi_{i,j}$.
	Sort the items in $D_i$ increasingly by their contribution to $i$'s utility: $\phi_{i,i_1} \leq \phi_{i,i_2} \leq \ldots \leq \phi_{i,i_{m_i}}$; it follows that $u_i(\eqbb) = \phi_{i,i_1}$. 
	Let $S_i = \{j \in D_i \; | \; \phi_{i,j} = \phi_{i,i_1} \}$ be the items received in the smallest fraction (equal to $i$'s utility). If $S_i = M$, then the analysis is similar to the exact equilibrium case, where the prices are strictly positive. The difficult case is when 
	$S_i \neq M$. Then player $i$ is getting a higher than necessary fraction from some resource $j \in M \setminus S_i$. Thus $i$ would improve
	by shifting some of the mass from item $j$ to the items in $S$. Since $\eqbb$ is an equilibrium, no such deviation is possible. Then it must be the case that $\tilde{b}_{i,j} = \Delta$ for all $j \in D_i \setminus S_i$. 
	
	Now interpret the bids and allocation as a market equilibrium with Leontief utilities $\vec{v}$ and budgets $B_i$, by setting the prices to
	$\vec{p}= (p_1, \ldots, p_m)$, where $p_j = \sum_{i=1}^{n} \tilde{b}_{i,j}$ for all $j \in M$, and the allocation to $\vec{x}$, the same as the one induced by the bids $\eqbb$ in the trading post game. 
	We argue that $(\vec{p}, \vec{x})$ is an $\epsilon$-market equilibrium. Clearly at the outcome $(\vec{p}, \vec{x})$ all 
	the goods are sold and each buyer exhausts their budget. Moreover observe that all the prices are strictly positive. 
	We must additionally show that each player gets an $\epsilon$-optimal bundle at $(\vec{p}, \vec{x})$.
	
	Fix an arbitrary buyer $i$. Let $\vec{y}_i$ be an optimal bundle for $i$ given prices $\vec{p}$, and
	let $q_{i,j}$ be the amount of money spent by $i$ to purchase $y_{i,j}$ units of good $j$ at these prices.
	An upper bound on the optimal value $u_i(\vec{y}_i)$ is attained when buyer $i$ shifts \emph{all} the 
	money spent on purchasing items outside $S_i$ to purchase instead higher fractions from the items in $S_i$. 
	Since the strategy profile $\vec{b}$ is an exact equilibrium in the game $\mathcal{TP}(\Delta)$, the amount of money spent by player $i$ on items outside $S_i$ is at most $(m-1) \Delta$; thus $i$ spends at most $B_i - (m-1)\Delta$ on the remaining items in $S_i$.
	
	By an averaging argument, there exists a good $j \in S_i$ on which $i$ spends the greatest amount of its money, i.e. $$\tilde{b}_{i,j} \geq \frac{B_i - (m-1)\Delta}{|S_i|}.$$ This will be the item for which the gain brought by the deviation in spending is modest.
	Formally, the maximum fraction of utility that $i$ can get from item $j$---without decreasing the ratios at which the other items in $S_i$ are received---is:
	\begin{eqnarray*}
		\phi'_{i,j}& =& \frac{q_{i,j}}{p_{j}\cdot v_{i,j}} 
		\leq  \frac{\tilde{b}_{i,j} + (m-1)\Delta}{\left(\sum_{k=1}^{n} \tilde{b}_{k,j} \right) \cdot v_{i,j}} 
		=  \phi_{i,j} +  \frac{(m-1)\Delta}{\left( \sum_{k=1}^{n} \tilde{b}_{k,j} \right) \cdot v_{i,j}} \\
		&\leq&  \phi_{i,j} +  \frac{\tilde{b}_{i,j} \cdot \epsilon}{ \left(\sum_{k =1}^{n} \tilde{b}_{k,j} \right) \cdot v_{i,j}} = \phi_{i,j} \cdot (1 + \epsilon) 
		=  u_i(\vec{x}_i) (1 + \epsilon)
	\end{eqnarray*}
	where in the inequalities we additionally used that $\Delta < \epsilon^2/m$, $B_i \geq 1 \; \forall i \in N$, and $S_i \leq m-1$. The identities hold since item $j$ is in the tight set $S_i$.
	Then $u_i(\vec{y}_i) \leq \phi'_{i,j} \leq u_i(\vec{x}_i) (1 + \epsilon)$.
	Thus each player gets an $\epsilon$-optimal bundle, and so $(\vec{p}, \vec{x})$ is an $\epsilon$-market equilibrium.
\end{proof}

The following theorem, which we believe is of independent interest, states that
in Fisher markets with Leontief utilities, approximate market equilibria are
close to exact equilibria in terms of their Nash Social Welfare. 

\begin{theorem} \label{thm:apx_market}
	The Nash social welfare at an $\epsilon$-market equilibrium for Leontief utilities is at least
	a $\frac{1}{(1+\epsilon)}$ factor of the optimal Nash social welfare. 
\end{theorem}

Finally, we can state the main result of this section.

\begin{theorem}
	For every $\epsilon > 0$, the Trading Post game $\mathcal{TP}(\Delta)$ with entrance fee 
	$0 < \Delta < \min \left\{\frac{\epsilon^2}{m}, \frac{1}{m}\right\}$ has a price of anarchy of $1 + \epsilon$,
	even for arbitrary budgets.
\end{theorem}
\begin{proof}
	By Theorem~\ref{thm:TP_apx_market}, every PNE of $\mathcal{TP}(\Delta)$ corresponds to an $\epsilon$-market 
	equilibrium. By Theorem~\ref{thm:apx_market}, every $\epsilon$-market equilibrium attains at least a fraction
	$\frac{1}{1+\epsilon}$ of the optimal Nash social welfare. Thus, the price of anarchy of $\mathcal{TP}(\Delta)$ 
	is $1 + \epsilon$, which completes the proof. 
\end{proof}

\subsection{Trading Post: Beyond Perfect Substitutes and Complements}
\rn{
	In general market equilibrium theory the valuations of the agents are considered to be arbitrary concave, non-decreasing and non-negative functions.
	Next we extend the upper bound of $2$ achieved by Trading Post to this most general setting. 
	The missing proofs are in the full version \cite{fullversion}.}

\begin{theorem}\label{thm:ub_tp_con}
	The Trading Post game with concave valuations has price of anarchy at most 2.
\end{theorem}

Moreover, the existence of pure Nash equilibria holds for all CES utilities.

\begin{theorem} \label{thm:existence_CES}
	The Trading Post game with no minimum bid has exact pure Nash equilibria for all CES utilities with perfect competition and $\rho \in (-\infty, 1]$.
\end{theorem}

\subsection{Trading Post: Beyond Pure Nash Equilibria}

A natural remaining question is
whether there exist mixed Nash equilibria with bad Nash social welfare. In the case of Trading Post mechanism we rule out this case by showing that all its 
Nash equilibria are pure, be it for linear, Leontief, or concave valuations. 
Details of this result may be found in the full version \cite{fullversion}, where we
show the following theorem.

\begin{theorem}\label{thm:nomixed}
	For markets with linear utilities, every Nash equilibrium of the corresponding Trading Post game is pure. \rn{Further, the result extends to concave valuations if a mild condition of {\em enough competition} is satisfied .}
	For markets with Leontief utilities, every Nash equilibrium of the corresponding $\Delta$ Trading Post game $TP(\Delta)$, for
	$\Delta>0$, is pure.
\end{theorem}

Given such an efficiency achieved by Trading Post mechanism one may wonder if these mechanisms in-fact have unique equilibrium.
We rule this out through constructions for additive and Leontief valuations, given in the full version \cite{fullversion}. 

\section{Fairness Guarantees}
Finally, for concave utilities all the Nash equilibria of the two mechanisms satisfy an important notion of fairness: proportionality. That is, each player $i$ gets a fraction of at least $B_i/\mathcal{B}$ of its maximum utility (in any Nash equilibrium). For equal budgets, the guarantee is $1/n$.

\begin{theorem} \label{thm:TP_fairness}
	Every Nash equilibrium of the Fisher market mechanism with concave utilities, and of the Trading Post mechanism with concave and strictly increasing utilities 
	is proportional. Moreover, for every $0 < \Delta < 1/m$, every Nash equilibrium of the parameterized Trading Post game $TP(\Delta)$ with concave and strictly increasing utilities is approximately proportional, guaranteeing at least $\frac{B_i}{\mathcal{B}}\left(1 - \rho_i\right)$ of the optimum for each player $i$, where $\rho_i = \frac{\Delta \cdot(m-1)}{B_i}$, $B_i$ is the player's budget, $\mathcal{B}$ the sum of budgets, and $m$ the number of items.
\end{theorem}

\section{Discussion}\label{app:disc}
Our results show great differences between the Fisher market mechanism and Trading Post, with Trading Post simultaneously achieving a Nash social welfare within a factor of two of the optimum and fairness (in the form of proportionality) for all concave utilities. Existence of pure Nash equilibria in Trading Post is guaranteed for all CES utilities.

Our most unexpected and important contribution is for Leontief utilities, where Trading Post approximates the NSW arbitrarily close, while still guaranteeing fair outcomes to each individual. This result may have implications beyond theory, since Leontief utilities have played a starring role in the emerging literature on allocating computer resources. One of the most influential papers in this literature \cite{GZH11} suggested that the NSW maximizing outcome would be the ideal solution, its only drawback being the fact that it cannot be reached in the presence of strategic participants.

An interesting open question is that of achieving good approximations of the NSW objective using truthful,
but non-wasteful, mechanisms. 

\section{Acknowledgements}
Simina was supported by ISF grant 1435/14 administered by the Israeli Academy of Sciences and Israel-USA 
Bi-national Science Foundation (BSF) grant  2014389, and the I-CORE Program of the Planning 
and Budgeting Committee and The Israel Science Foundation. 
This project has received funding from the European Research Council (ERC) under the European Union's
Horizon 2020 research and innovation programme (grant agreement No 740282).
Vasilis was supported by NSF grants 
CCF-1408635, CCF-1216073, and CCF-1161813. This work was done in part 
while the authors were research fellows at the Simons Institute for the Theory of Computing.


\bibliographystyle{ACM-Reference-Format}
\bibliography{RSV}


\begin{thebibliography}{00}


\ifx \showCODEN    \undefined \def \showCODEN     #1{\unskip}     \fi
\ifx \showDOI      \undefined \def \showDOI       #1{{\tt DOI:}\penalty0{#1}\ }
  \fi
\ifx \showISBNx    \undefined \def \showISBNx     #1{\unskip}     \fi
\ifx \showISBNxiii \undefined \def \showISBNxiii  #1{\unskip}     \fi
\ifx \showISSN     \undefined \def \showISSN      #1{\unskip}     \fi
\ifx \showLCCN     \undefined \def \showLCCN      #1{\unskip}     \fi
\ifx \shownote     \undefined \def \shownote      #1{#1}          \fi
\ifx \showarticletitle \undefined \def \showarticletitle #1{#1}   \fi
\ifx \showURL      \undefined \def \showURL       #1{#1}          \fi
\providecommand\bibfield[2]{#2}
\providecommand\bibinfo[2]{#2}
\providecommand\natexlab[1]{#1}
\providecommand\showeprint[2][]{arXiv:#2}

\bibitem[\protect\citeauthoryear{Adsul, Babu, Garg, Mehta, and Sohoni}{Adsul
  et~al\mbox{.}}{2010}]%
        {ABGMS}
\bibfield{author}{\bibinfo{person}{B. Adsul}, \bibinfo{person}{Ch.~Sobhan
  Babu}, \bibinfo{person}{J. Garg}, \bibinfo{person}{R. Mehta}, {and}
  \bibinfo{person}{M. Sohoni}.} \bibinfo{year}{2010}\natexlab{}.
\newblock \showarticletitle{{N}ash equilibria in {F}isher market}. In
  \bibinfo{booktitle}{{\em SAGT}}. \bibinfo{pages}{30--41}.
\newblock


\bibitem[\protect\citeauthoryear{{Anari}, {Mai}, {Oveis Gharan}, and
  {Vazirani}}{{Anari} et~al\mbox{.}}{2016}]%
        {Anari2}
\bibfield{author}{\bibinfo{person}{N. {Anari}}, \bibinfo{person}{T. {Mai}},
  \bibinfo{person}{S. {Oveis Gharan}}, {and} \bibinfo{person}{V.~V.
  {Vazirani}}.} \bibinfo{year}{2016}\natexlab{}.
\newblock \showarticletitle{{Nash Social Welfare for Indivisible Items under
  Separable, Piecewise-Linear Concave Utilities}}.
\newblock \bibinfo{journal}{{\em ArXiv e-prints\/}} (\bibinfo{date}{Dec.}
  \bibinfo{year}{2016}).
\newblock
\showeprint[arxiv]{cs.GT/1612.05191}


\bibitem[\protect\citeauthoryear{Anari, Oveis-Gharan, Saberi, and Singh}{Anari
  et~al\mbox{.}}{2017}]%
        {Anari1}
\bibfield{author}{\bibinfo{person}{N. Anari}, \bibinfo{person}{S.
  Oveis-Gharan}, \bibinfo{person}{A. Saberi}, {and} \bibinfo{person}{M.
  Singh}.} \bibinfo{year}{2017}\natexlab{}.
\newblock \showarticletitle{Nash social welfare, matrix permanent, and stable
  polynomials}. In \bibinfo{booktitle}{{\em ITCS}}.
\newblock


\bibitem[\protect\citeauthoryear{Arrow and Debreu}{Arrow and Debreu}{1954}]%
        {AD54}
\bibfield{author}{\bibinfo{person}{K.~J. Arrow} {and} \bibinfo{person}{G.
  Debreu}.} \bibinfo{year}{1954}\natexlab{}.
\newblock \showarticletitle{Existence of an Equilibrium for a Competitive
  Economy}.
\newblock \bibinfo{journal}{{\em Econometrica\/}} \bibinfo{number}{22}
  (\bibinfo{year}{1954}), \bibinfo{pages}{265--290}.
\newblock


\bibitem[\protect\citeauthoryear{Aziz and Mackenzie}{Aziz and
  Mackenzie}{2016}]%
        {AM16}
\bibfield{author}{\bibinfo{person}{H. Aziz} {and} \bibinfo{person}{S.
  Mackenzie}.} \bibinfo{year}{2016}\natexlab{}.
\newblock \showarticletitle{A Discrete and Bounded Envy-free Cake Cutting
  Protocol for Four Agents}. In \bibinfo{booktitle}{{\em STOC}}.
\newblock


\bibitem[\protect\citeauthoryear{Babaioff, Lucier, Nisan, and
  Paes~Leme}{Babaioff et~al\mbox{.}}{2014}]%
        {BBNR}
\bibfield{author}{\bibinfo{person}{M. Babaioff}, \bibinfo{person}{B. Lucier},
  \bibinfo{person}{N. Nisan}, {and} \bibinfo{person}{R. Paes~Leme}.}
  \bibinfo{year}{2014}\natexlab{}.
\newblock \showarticletitle{On the Efficiency of the {W}alrasian Mechanism}. In
  \bibinfo{booktitle}{{\em EC}}. \bibinfo{pages}{783--800}.
\newblock


\bibitem[\protect\citeauthoryear{Barbanel}{Barbanel}{2004}]%
        {Barbanel04}
\bibfield{author}{\bibinfo{person}{J.B. Barbanel}.}
  \bibinfo{year}{2004}\natexlab{}.
\newblock \bibinfo{booktitle}{{\em The Geometry of Efficient Fair Division}}.
\newblock \bibinfo{publisher}{Cambridge Univ. Press}.
\newblock


\bibitem[\protect\citeauthoryear{Bevia, Corch\'{o}n, and Wilkie}{Bevia
  et~al\mbox{.}}{2003}]%
        {BCW}
\bibfield{author}{\bibinfo{person}{C. Bevia}, \bibinfo{person}{L. Corch\'{o}n},
  {and} \bibinfo{person}{S. Wilkie}.} \bibinfo{year}{2003}\natexlab{}.
\newblock \showarticletitle{Implementation of the {W}alrasian correspondence by
  market games}.
\newblock \bibinfo{journal}{{\em Review of Economic Design\/}}
  \bibinfo{volume}{7} (\bibinfo{year}{2003}), \bibinfo{pages}{429--442}.
\newblock


\bibitem[\protect\citeauthoryear{Brainard and Scarf}{Brainard and
  Scarf}{2000}]%
        {BSAD}
\bibfield{author}{\bibinfo{person}{W.~C. Brainard} {and} \bibinfo{person}{H.~E.
  Scarf}.} \bibinfo{year}{2000}\natexlab{}.
\newblock \showarticletitle{How to compute equilibrium prices in 1891}.
\newblock \bibinfo{journal}{{\em Cowles Foundation Discussion Paper\/}}
  \bibinfo{volume}{1270} (\bibinfo{year}{2000}).
\newblock


\bibitem[\protect\citeauthoryear{Brams and Taylor}{Brams and Taylor}{1996}]%
        {BramsT96}
\bibfield{author}{\bibinfo{person}{S. Brams} {and} \bibinfo{person}{A.
  Taylor}.} \bibinfo{year}{1996}\natexlab{}.
\newblock \bibinfo{booktitle}{{\em Fair Division: from cake cutting to dispute
  resolution}}.
\newblock \bibinfo{publisher}{Cambridge University Press, Cambridge}.
\newblock


\bibitem[\protect\citeauthoryear{Brandt, Conitzer, Endriss, Lang, and
  Procaccia}{Brandt et~al\mbox{.}}{2016}]%
        {COMSOC}
\bibfield{author}{\bibinfo{person}{F. Brandt}, \bibinfo{person}{V. Conitzer},
  \bibinfo{person}{U. Endriss}, \bibinfo{person}{J. Lang}, {and}
  \bibinfo{person}{A.~D. Procaccia}.} \bibinfo{year}{2016}\natexlab{}.
\newblock \bibinfo{booktitle}{{\em Handbook of Computational Social Choice}}.
\newblock \bibinfo{publisher}{Cambridge University Press}.
\newblock


\bibitem[\protect\citeauthoryear{Br{\^{a}}nzei, Caragiannis, Kurokawa, and
  Procaccia}{Br{\^{a}}nzei et~al\mbox{.}}{2016}]%
        {BCKP16}
\bibfield{author}{\bibinfo{person}{S. Br{\^{a}}nzei}, \bibinfo{person}{I.
  Caragiannis}, \bibinfo{person}{D. Kurokawa}, {and} \bibinfo{person}{A.D.
  Procaccia}.} \bibinfo{year}{2016}\natexlab{}.
\newblock \showarticletitle{{An Algorithmic Framework for Strategic Fair
  Division}}. In \bibinfo{booktitle}{{\em AAAI}}.
\newblock


\bibitem[\protect\citeauthoryear{Br{\^a}nzei, Chen, Deng, Filos-Ratsikas,
  Frederiksen, and Zhang}{Br{\^a}nzei et~al\mbox{.}}{2014}]%
        {BCDFFZ14}
\bibfield{author}{\bibinfo{person}{S. Br{\^a}nzei}, \bibinfo{person}{Y. Chen},
  \bibinfo{person}{X. Deng}, \bibinfo{person}{A. Filos-Ratsikas},
  \bibinfo{person}{S. Frederiksen}, {and} \bibinfo{person}{J. Zhang}.}
  \bibinfo{year}{2014}\natexlab{}.
\newblock \showarticletitle{The {F}isher Market Game: Equilibrium and Welfare}.
  In \bibinfo{booktitle}{{\em AAAI}}. \bibinfo{pages}{587--593}.
\newblock


\bibitem[\protect\citeauthoryear{Branzei, Gkatzelis, and Mehta}{Branzei
  et~al\mbox{.}}{2017}]%
        {fullversion}
\bibfield{author}{\bibinfo{person}{S. Branzei}, \bibinfo{person}{V. Gkatzelis},
  {and} \bibinfo{person}{R. Mehta}.} \bibinfo{year}{2017}\natexlab{}.
\newblock \bibinfo{title}{Nash Social Welfare Approximation for Strategic
  Agents}.  (\bibinfo{year}{2017}).
\newblock
\newblock
\shownote{arxiv/abs/1607.01569.}


\bibitem[\protect\citeauthoryear{Br{\^{a}}nzei and Miltersen}{Br{\^{a}}nzei and
  Miltersen}{2015}]%
        {BM15}
\bibfield{author}{\bibinfo{person}{S. Br{\^{a}}nzei} {and}
  \bibinfo{person}{P.~B. Miltersen}.} \bibinfo{year}{2015}\natexlab{}.
\newblock \showarticletitle{A Dictatorship Theorem for Cake Cutting}. In
  \bibinfo{booktitle}{{\em IJCAI}}. \bibinfo{pages}{482--488}.
\newblock


\bibitem[\protect\citeauthoryear{Cambini and Martein}{Cambini and
  Martein}{2009}]%
        {CM09}
\bibfield{author}{\bibinfo{person}{A. Cambini} {and} \bibinfo{person}{L.
  Martein}.} \bibinfo{year}{2009}\natexlab{}.
\newblock \bibinfo{booktitle}{{\em {Generalized Convexity and Optimization:
  Theory and Applications}}}.
\newblock \bibinfo{publisher}{Springer}.
\newblock


\bibitem[\protect\citeauthoryear{Caragiannis, Kurokawa, Moulin, Procaccia,
  Shah, and Wang}{Caragiannis et~al\mbox{.}}{2016}]%
        {rCKMPSW16}
\bibfield{author}{\bibinfo{person}{I. Caragiannis}, \bibinfo{person}{D.
  Kurokawa}, \bibinfo{person}{H.~C. Moulin}, \bibinfo{person}{A.~D. Procaccia},
  \bibinfo{person}{N. Shah}, {and} \bibinfo{person}{J. Wang}.}
  \bibinfo{year}{2016}\natexlab{}.
\newblock \showarticletitle{The Unreasonable Fairness of Maximum Nash Welfare}.
  In \bibinfo{booktitle}{{\em EC}}. \bibinfo{pages}{305--322}.
\newblock


\bibitem[\protect\citeauthoryear{Chakrabarty, Chuzhoy, and Khanna}{Chakrabarty
  et~al\mbox{.}}{2009}]%
        {CCK09}
\bibfield{author}{\bibinfo{person}{D. Chakrabarty}, \bibinfo{person}{J.
  Chuzhoy}, {and} \bibinfo{person}{S. Khanna}.}
  \bibinfo{year}{2009}\natexlab{}.
\newblock \showarticletitle{On Allocating Goods to Maximize Fairness}. In
  \bibinfo{booktitle}{{\em FOCS}}. \bibinfo{pages}{107--116}.
\newblock


\bibitem[\protect\citeauthoryear{Chen, Deng, Sun, and Yao}{Chen
  et~al\mbox{.}}{2004}]%
        {CDSYao04}
\bibfield{author}{\bibinfo{person}{N. Chen}, \bibinfo{person}{Xiatie Deng},
  \bibinfo{person}{Xiaoming Sun}, {and} \bibinfo{person}{Andrew Yao}.}
  \bibinfo{year}{2004}\natexlab{}.
\newblock \showarticletitle{{{F}isher Equilibrium Price with a class of Concave
  Utility Functions}}. In \bibinfo{booktitle}{{\em ESA}}.
  \bibinfo{pages}{169--179}.
\newblock


\bibitem[\protect\citeauthoryear{Chen, Deng, Zhang, and Zhang}{Chen
  et~al\mbox{.}}{2012}]%
        {CDZZ12}
\bibfield{author}{\bibinfo{person}{N. Chen}, \bibinfo{person}{X. Deng},
  \bibinfo{person}{H. Zhang}, {and} \bibinfo{person}{J. Zhang}.}
  \bibinfo{year}{2012}\natexlab{}.
\newblock \showarticletitle{Incentive Ratios of {F}isher Markets}. In
  \bibinfo{booktitle}{{\em ICALP}}.
\newblock


\bibitem[\protect\citeauthoryear{Chen, Deng, and Zhang}{Chen
  et~al\mbox{.}}{2011}]%
        {CDZ11}
\bibfield{author}{\bibinfo{person}{N. Chen}, \bibinfo{person}{X. Deng}, {and}
  \bibinfo{person}{J. Zhang}.} \bibinfo{year}{2011}\natexlab{}.
\newblock \showarticletitle{How Profitable Are Strategic Behaviors in a
  Market?}. In \bibinfo{booktitle}{{\em ESA}}.
\newblock


\bibitem[\protect\citeauthoryear{Chen, Lai, Parkes, and Procaccia}{Chen
  et~al\mbox{.}}{2013}]%
        {CLPP10}
\bibfield{author}{\bibinfo{person}{Y. Chen}, \bibinfo{person}{J.K. Lai},
  \bibinfo{person}{D.C. Parkes}, {and} \bibinfo{person}{A.D. Procaccia}.}
  \bibinfo{year}{2013}\natexlab{}.
\newblock \showarticletitle{Truth, justice, and cake cutting}.
\newblock \bibinfo{journal}{{\em Games and Economic Behavior\/}}
  \bibinfo{volume}{77}, \bibinfo{number}{1} (\bibinfo{year}{2013}),
  \bibinfo{pages}{284--297}.
\newblock


\bibitem[\protect\citeauthoryear{Codenotti and Vardarajan}{Codenotti and
  Vardarajan}{2004}]%
        {BrunoVICALP04}
\bibfield{author}{\bibinfo{person}{B. Codenotti} {and} \bibinfo{person}{K.
  Vardarajan}.} \bibinfo{year}{2004}\natexlab{}.
\newblock \showarticletitle{{Equilbrium for Markets with with Leontief
  Utilities}}. In \bibinfo{booktitle}{{\em ICALP}}.
\newblock


\bibitem[\protect\citeauthoryear{Cole, Devanur, Gkatzelis, Jain, Mai, Vazirani,
  and Yazdanbod}{Cole et~al\mbox{.}}{2017}]%
        {CDGJMVY17}
\bibfield{author}{\bibinfo{person}{R. Cole}, \bibinfo{person}{N.R. Devanur},
  \bibinfo{person}{V. Gkatzelis}, \bibinfo{person}{K. Jain},
  \bibinfo{person}{T. Mai}, \bibinfo{person}{V.V. Vazirani}, {and}
  \bibinfo{person}{S. Yazdanbod}.} \bibinfo{year}{2017}\natexlab{}.
\newblock \showarticletitle{Convex Program Duality, Fisher Markets, and Nash
  Social Welfare}. In \bibinfo{booktitle}{{\em EC}}.
\newblock


\bibitem[\protect\citeauthoryear{Cole and Gkatzelis}{Cole and
  Gkatzelis}{2015}]%
        {CG15}
\bibfield{author}{\bibinfo{person}{R. Cole} {and} \bibinfo{person}{V.
  Gkatzelis}.} \bibinfo{year}{2015}\natexlab{}.
\newblock \showarticletitle{{Approximating the Nash Social Welfare with
  Indivisible Items}}. In \bibinfo{booktitle}{{\em STOC}}.
  \bibinfo{pages}{371--380}.
\newblock


\bibitem[\protect\citeauthoryear{Cole, Gkatzelis, and Goel}{Cole
  et~al\mbox{.}}{2013a}]%
        {CGG13a}
\bibfield{author}{\bibinfo{person}{R. Cole}, \bibinfo{person}{V. Gkatzelis},
  {and} \bibinfo{person}{G. Goel}.} \bibinfo{year}{2013}\natexlab{a}.
\newblock \showarticletitle{Mechanism design for fair division: allocating
  divisible items without payments}. In \bibinfo{booktitle}{{\em EC}}.
  \bibinfo{pages}{251--268}.
\newblock


\bibitem[\protect\citeauthoryear{Cole, Gkatzelis, and Goel}{Cole
  et~al\mbox{.}}{2013b}]%
        {CGG13b}
\bibfield{author}{\bibinfo{person}{R. Cole}, \bibinfo{person}{V. Gkatzelis},
  {and} \bibinfo{person}{G. Goel}.} \bibinfo{year}{2013}\natexlab{b}.
\newblock \showarticletitle{Positive results for mechanism design without
  money}. In \bibinfo{booktitle}{{\em AAMAS}}. \bibinfo{pages}{1165--1166}.
\newblock


\bibitem[\protect\citeauthoryear{Cole and Tao}{Cole and Tao}{2016}]%
        {CT}
\bibfield{author}{\bibinfo{person}{R. Cole} {and} \bibinfo{person}{Y. Tao}.}
  \bibinfo{year}{2016}\natexlab{}.
\newblock \showarticletitle{Large Market Games with Near Optimal Efficiency}.
  In \bibinfo{booktitle}{{\em Proceedings of the 2016 {ACM} Conference on
  Economics and Computation, {EC} 2016,}}. \bibinfo{pages}{791--808}.
\newblock


\bibitem[\protect\citeauthoryear{Dasgupta, Hammond, and Maskin}{Dasgupta
  et~al\mbox{.}}{1979}]%
        {DHM}
\bibfield{author}{\bibinfo{person}{P. Dasgupta}, \bibinfo{person}{P. Hammond},
  {and} \bibinfo{person}{E. Maskin}.} \bibinfo{year}{1979}\natexlab{}.
\newblock \showarticletitle{The implementation of social choice rules: Some
  general results on incentive compatibility}.
\newblock \bibinfo{journal}{{\em The Review of Economic Studies\/}}
  \bibinfo{volume}{46}, \bibinfo{number}{2} (\bibinfo{year}{1979}),
  \bibinfo{pages}{185--216}.
\newblock


\bibitem[\protect\citeauthoryear{Debreu}{Debreu}{1952}]%
        {Debreu52}
\bibfield{author}{\bibinfo{person}{G. Debreu}.}
  \bibinfo{year}{1952}\natexlab{}.
\newblock \showarticletitle{A Social Equilibrium Existence Theorem}.
\newblock \bibinfo{journal}{{\em PNAS\/}}  \bibinfo{volume}{38}
  (\bibinfo{year}{1952}), \bibinfo{pages}{886--893}.
\newblock


\bibitem[\protect\citeauthoryear{Devanur}{Devanur}{2009}]%
        {dev}
\bibfield{author}{\bibinfo{person}{N.R. Devanur}.}
  \bibinfo{year}{2009}\natexlab{}.
\newblock \bibinfo{title}{{F}isher markets and convex programs}.
  (\bibinfo{year}{2009}).
\newblock
\newblock
\shownote{Manuscript.}


\bibitem[\protect\citeauthoryear{Devanur, Papadimitriou, Saberi, and
  Vazirani}{Devanur et~al\mbox{.}}{2008}]%
        {DPSV}
\bibfield{author}{\bibinfo{person}{N. Devanur}, \bibinfo{person}{C.H.
  Papadimitriou}, \bibinfo{person}{A. Saberi}, {and} \bibinfo{person}{V.~V.
  Vazirani}.} \bibinfo{year}{2008}\natexlab{}.
\newblock \showarticletitle{Market equilibrium via a primal-dual algorithm for
  a convex program}.
\newblock \bibinfo{journal}{{\em JACM\/}} \bibinfo{volume}{55},
  \bibinfo{number}{5} (\bibinfo{year}{2008}).
\newblock


\bibitem[\protect\citeauthoryear{Dolev, Feitelson, Halpern, Kupferman, and
  Linial}{Dolev et~al\mbox{.}}{2012}]%
        {DFH12}
\bibfield{author}{\bibinfo{person}{D. Dolev}, \bibinfo{person}{D.G. Feitelson},
  \bibinfo{person}{J.Y. Halpern}, \bibinfo{person}{R. Kupferman}, {and}
  \bibinfo{person}{N. Linial}.} \bibinfo{year}{2012}\natexlab{}.
\newblock \showarticletitle{No justified complaints: on fair sharing of
  multiple resources}. In \bibinfo{booktitle}{{\em ITCS}}.
  \bibinfo{pages}{68--75}.
\newblock


\bibitem[\protect\citeauthoryear{Dubey and Geanakoplos}{Dubey and
  Geanakoplos}{2003}]%
        {DG}
\bibfield{author}{\bibinfo{person}{P. Dubey} {and} \bibinfo{person}{J.
  Geanakoplos}.} \bibinfo{year}{2003}\natexlab{}.
\newblock \showarticletitle{From {N}ash to {W}alras via {S}hapley-{S}hubik}.
\newblock \bibinfo{journal}{{\em Journal of Mathematical Economics\/}}
  \bibinfo{volume}{39}, \bibinfo{number}{5} (\bibinfo{year}{2003}),
  \bibinfo{pages}{391--400}.
\newblock


\bibitem[\protect\citeauthoryear{Dubey and Shubik}{Dubey and Shubik}{1978}]%
        {DS}
\bibfield{author}{\bibinfo{person}{P. Dubey} {and} \bibinfo{person}{M.
  Shubik}.} \bibinfo{year}{1978}\natexlab{}.
\newblock \showarticletitle{The noncooperative equilibria of a closed trading
  economy with market supply and bidding strategies}.
\newblock \bibinfo{journal}{{\em JET\/}}  \bibinfo{volume}{17}
  (\bibinfo{year}{1978}), \bibinfo{pages}{1--20}.
\newblock


\bibitem[\protect\citeauthoryear{Eaves}{Eaves}{1976}]%
        {Ea}
\bibfield{author}{\bibinfo{person}{B.C. Eaves}.}
  \bibinfo{year}{1976}\natexlab{}.
\newblock \showarticletitle{A finite algorithm for the linear exchange model}.
\newblock \bibinfo{journal}{{\em J. Math. Econ.\/}}  \bibinfo{volume}{3}
  (\bibinfo{year}{1976}), \bibinfo{pages}{197--203}.
\newblock


\bibitem[\protect\citeauthoryear{Eisenberg and Gale}{Eisenberg and
  Gale}{1959}]%
        {EG}
\bibfield{author}{\bibinfo{person}{E. Eisenberg} {and} \bibinfo{person}{D.
  Gale}.} \bibinfo{year}{1959}\natexlab{}.
\newblock \showarticletitle{Consensus of subjective probabilities: the
  {P}ari-{M}utuel method}.
\newblock \bibinfo{journal}{{\em The Annals of Mathematical Statistics\/}}
  \bibinfo{volume}{30} (\bibinfo{year}{1959}), \bibinfo{pages}{165--168}.
\newblock


\bibitem[\protect\citeauthoryear{Fang}{Fang}{2002}]%
        {Fang02}
\bibfield{author}{\bibinfo{person}{H. Fang}.} \bibinfo{year}{2002}\natexlab{}.
\newblock \showarticletitle{Lottery versus all-pay auction models of lobbying}.
\newblock \bibinfo{journal}{{\em Public Choice\/}} \bibinfo{volume}{112},
  \bibinfo{number}{3-4} (\bibinfo{year}{2002}), \bibinfo{pages}{351--71}.
\newblock


\bibitem[\protect\citeauthoryear{Feldman, Lai, and Zhang}{Feldman
  et~al\mbox{.}}{2009}]%
        {FLZ09}
\bibfield{author}{\bibinfo{person}{M. Feldman}, \bibinfo{person}{K. Lai}, {and}
  \bibinfo{person}{L. Zhang}.} \bibinfo{year}{2009}\natexlab{}.
\newblock \showarticletitle{The Proportional-Share Allocation Market for
  Computational Resources}.
\newblock \bibinfo{journal}{{\em ITPDS\/}} \bibinfo{volume}{20},
  \bibinfo{number}{8} (\bibinfo{year}{2009}).
\newblock


\bibitem[\protect\citeauthoryear{Gale}{Gale}{1960}]%
        {gale}
\bibfield{author}{\bibinfo{person}{D. Gale}.} \bibinfo{year}{1960}\natexlab{}.
\newblock \bibinfo{booktitle}{{\em Theory of Linear Economic Models}}.
\newblock \bibinfo{publisher}{McGraw Hill}, \bibinfo{address}{N.Y.}
\newblock


\bibitem[\protect\citeauthoryear{Gale}{Gale}{1976}]%
        {Ga76}
\bibfield{author}{\bibinfo{person}{D. Gale}.} \bibinfo{year}{1976}\natexlab{}.
\newblock \showarticletitle{The Linear exchange Model}.
\newblock \bibinfo{journal}{{\em J. Math. Econ.\/}}  \bibinfo{volume}{3}
  (\bibinfo{year}{1976}), \bibinfo{pages}{205--209}.
\newblock


\bibitem[\protect\citeauthoryear{Garg, Kapoor, and Vazirani}{Garg
  et~al\mbox{.}}{2004}]%
        {auction.gross}
\bibfield{author}{\bibinfo{person}{R. Garg}, \bibinfo{person}{S. Kapoor}, {and}
  \bibinfo{person}{V.~V. Vazirani}.} \bibinfo{year}{2004}\natexlab{}.
\newblock \showarticletitle{An Auction-Based Market Equilbrium Algorithm for
  the Separable Gross Substitutibility Case}. In \bibinfo{booktitle}{{\em
  APPROX}}.
\newblock


\bibitem[\protect\citeauthoryear{Georgiou, Pavlides, and Philippou}{Georgiou
  et~al\mbox{.}}{2006}]%
        {GPP06}
\bibfield{author}{\bibinfo{person}{C. Georgiou}, \bibinfo{person}{T. Pavlides},
  {and} \bibinfo{person}{A. Philippou}.} \bibinfo{year}{2006}\natexlab{}.
\newblock \showarticletitle{Network uncertainty in selfish routing}. In
  \bibinfo{booktitle}{{\em International Parallel and Distributed Processing
  Symposium}}.
\newblock


\bibitem[\protect\citeauthoryear{Ghodsi, Zaharia, Hindman, Konwinski, Shenker,
  and Stoica}{Ghodsi et~al\mbox{.}}{2011}]%
        {GZH11}
\bibfield{author}{\bibinfo{person}{A. Ghodsi}, \bibinfo{person}{M. Zaharia},
  \bibinfo{person}{B. Hindman}, \bibinfo{person}{A. Konwinski},
  \bibinfo{person}{S. Shenker}, {and} \bibinfo{person}{I. Stoica}.}
  \bibinfo{year}{2011}\natexlab{}.
\newblock \showarticletitle{Dominant resource fairness: fair allocation of
  multiple resource types}. In \bibinfo{booktitle}{{\em NSDI}}.
\newblock


\bibitem[\protect\citeauthoryear{Giraud}{Giraud}{2003}]%
        {GG}
\bibfield{author}{\bibinfo{person}{G. Giraud}.}
  \bibinfo{year}{2003}\natexlab{}.
\newblock \showarticletitle{Strategic market games: an introduction}.
\newblock \bibinfo{journal}{{\em Journal of Mathematical Economics\/}}
  \bibinfo{volume}{39} (\bibinfo{year}{2003}), \bibinfo{pages}{355--375}.
\newblock


\bibitem[\protect\citeauthoryear{Glicksberg}{Glicksberg}{1952}]%
        {Glicksberg52}
\bibfield{author}{\bibinfo{person}{I.~L. Glicksberg}.}
  \bibinfo{year}{1952}\natexlab{}.
\newblock \showarticletitle{A Further Generalization of the Kakutani
  Fixed-Point Theorem}.
\newblock \bibinfo{journal}{{\it Proc. Amer. Math. Soc.}}  \bibinfo{volume}{3}
  (\bibinfo{year}{1952}), \bibinfo{pages}{170--174}.
\newblock


\bibitem[\protect\citeauthoryear{Gutman and Nisan}{Gutman and Nisan}{2012}]%
        {GN12}
\bibfield{author}{\bibinfo{person}{A. Gutman} {and} \bibinfo{person}{N.
  Nisan}.} \bibinfo{year}{2012}\natexlab{}.
\newblock \showarticletitle{Fair allocation without trade}. In
  \bibinfo{booktitle}{{\em AAMAS}}. \bibinfo{pages}{719--728}.
\newblock


\bibitem[\protect\citeauthoryear{Jackson and Peck}{Jackson and Peck}{1999}]%
        {JP}
\bibfield{author}{\bibinfo{person}{M.~O. Jackson} {and} \bibinfo{person}{J.
  Peck}.} \bibinfo{year}{1999}\natexlab{}.
\newblock \showarticletitle{Asymmetric information in a competitive market
  game: Reexamining the implications of rational expectations}.
\newblock \bibinfo{journal}{{\em Econ. Theory\/}}  \bibinfo{volume}{13}
  (\bibinfo{year}{1999}), \bibinfo{pages}{603--628}.
\newblock


\bibitem[\protect\citeauthoryear{Jain, Vazirani, and Ye}{Jain
  et~al\mbox{.}}{2005}]%
        {homothetic}
\bibfield{author}{\bibinfo{person}{K. Jain}, \bibinfo{person}{V.~V. Vazirani},
  {and} \bibinfo{person}{Y. Ye}.} \bibinfo{year}{2005}\natexlab{}.
\newblock \showarticletitle{Market Equilibrium for Homothetic, Quasi-Concave
  Utilities and Economies of Scale in Production}. In \bibinfo{booktitle}{{\em
  SODA}}.
\newblock


\bibitem[\protect\citeauthoryear{Kaneko and Nakamura}{Kaneko and
  Nakamura}{1979}]%
        {KN79}
\bibfield{author}{\bibinfo{person}{M. Kaneko} {and} \bibinfo{person}{K.
  Nakamura}.} \bibinfo{year}{1979}\natexlab{}.
\newblock \showarticletitle{{The {N}ash Social Welfare Function}}.
\newblock \bibinfo{journal}{{\em Econometrica\/}} \bibinfo{volume}{47},
  \bibinfo{number}{2} (\bibinfo{year}{1979}), \bibinfo{pages}{423--435}.
\newblock


\bibitem[\protect\citeauthoryear{Korpeoglu and Spear}{Korpeoglu and
  Spear}{2015}]%
        {KS}
\bibfield{author}{\bibinfo{person}{C.~G. Korpeoglu} {and}
  \bibinfo{person}{S.~E. Spear}.} \bibinfo{year}{2015}\natexlab{}.
\newblock \bibinfo{title}{The Market Game with Production: Coordination
  Equilibrium and Price Stickiness}.  (\bibinfo{year}{2015}).
\newblock
\newblock
\shownote{CMU, Tepper School of Business.}


\bibitem[\protect\citeauthoryear{Matros}{Matros}{2007}]%
        {Matros07}
\bibfield{author}{\bibinfo{person}{A. Matros}.}
  \bibinfo{year}{2007}\natexlab{}.
\newblock \bibinfo{title}{Chinese auctions}.
\newblock   (\bibinfo{year}{2007}).
\newblock
\newblock
\shownote{Mimeo, University of Pittsburgh.}


\bibitem[\protect\citeauthoryear{Maya and Nisan}{Maya and Nisan}{2012}]%
        {MN12}
\bibfield{author}{\bibinfo{person}{A. Maya} {and} \bibinfo{person}{N. Nisan}.}
  \bibinfo{year}{2012}\natexlab{}.
\newblock \showarticletitle{Incentive Compatible Two Player Cake Cutting}. In
  \bibinfo{booktitle}{{\em WINE}}.
\newblock


\bibitem[\protect\citeauthoryear{Mertens and (eds)}{Mertens and (eds)}{2013}]%
        {MS}
\bibfield{author}{\bibinfo{person}{J.-F. Mertens} {and}
  \bibinfo{person}{S.~Sorin (eds)}.} \bibinfo{year}{2013}\natexlab{}.
\newblock \bibinfo{booktitle}{{\em Game-theoretic methods in general
  equilibrium analysis}}.
\newblock \bibinfo{publisher}{Springer Science \& Business Media}.
\newblock


\bibitem[\protect\citeauthoryear{Moldovanu and Sela}{Moldovanu and
  Sela}{2001}]%
        {Moldovanu01}
\bibfield{author}{\bibinfo{person}{B. Moldovanu} {and} \bibinfo{person}{A.
  Sela}.} \bibinfo{year}{2001}\natexlab{}.
\newblock \showarticletitle{The optimal allocation of prizes in contests}.
\newblock \bibinfo{journal}{{\em AER\/}} \bibinfo{volume}{91},
  \bibinfo{number}{3} (\bibinfo{year}{2001}), \bibinfo{pages}{542--558}.
\newblock


\bibitem[\protect\citeauthoryear{Mossel and Tamuz}{Mossel and Tamuz}{2010}]%
        {MosselT10}
\bibfield{author}{\bibinfo{person}{E. Mossel} {and} \bibinfo{person}{O.
  Tamuz}.} \bibinfo{year}{2010}\natexlab{}.
\newblock \showarticletitle{Truthful Fair Division}. In
  \bibinfo{booktitle}{{\em SAGT}}. \bibinfo{pages}{288--299}.
\newblock


\bibitem[\protect\citeauthoryear{Moulin}{Moulin}{2003}]%
        {Moulin03}
\bibfield{author}{\bibinfo{person}{H. Moulin}.}
  \bibinfo{year}{2003}\natexlab{}.
\newblock \bibinfo{booktitle}{{\em Fair Division and Collective Welfare}}.
\newblock \bibinfo{publisher}{The MIT Press}.
\newblock


\bibitem[\protect\citeauthoryear{Nakamura}{Nakamura}{1990}]%
        {SN}
\bibfield{author}{\bibinfo{person}{S. Nakamura}.}
  \bibinfo{year}{1990}\natexlab{}.
\newblock \showarticletitle{A feasible {N}ash implementation of {W}alrasian
  equilibria in the two-agent economy}.
\newblock \bibinfo{journal}{{\em Economics Letters\/}} \bibinfo{volume}{34},
  \bibinfo{number}{1} (\bibinfo{year}{1990}), \bibinfo{pages}{5--9}.
\newblock


\bibitem[\protect\citeauthoryear{Nash}{Nash}{1950}]%
        {Nash50}
\bibfield{author}{\bibinfo{person}{J. Nash}.} \bibinfo{year}{1950}\natexlab{}.
\newblock \showarticletitle{The Bargaining Problem}.
\newblock \bibinfo{journal}{{\em Econometrica\/}} \bibinfo{volume}{18},
  \bibinfo{number}{2} (\bibinfo{date}{April} \bibinfo{year}{1950}),
  \bibinfo{pages}{155--162}.
\newblock


\bibitem[\protect\citeauthoryear{Nisan, Roughgarden, Tardos, and
  Vazirani}{Nisan et~al\mbox{.}}{2007}]%
        {NRTV07}
\bibfield{author}{\bibinfo{person}{N. Nisan}, \bibinfo{person}{T. Roughgarden},
  \bibinfo{person}{E. Tardos}, {and} \bibinfo{person}{V. Vazirani}.}
  \bibinfo{year}{2007}\natexlab{}.
\newblock \bibinfo{booktitle}{{\em Algorithmic Game Theory}}.
\newblock \bibinfo{publisher}{Cambridge University Press}.
\newblock


\bibitem[\protect\citeauthoryear{Orlin}{Orlin}{2010}]%
        {orlin}
\bibfield{author}{\bibinfo{person}{J.~B. Orlin}.}
  \bibinfo{year}{2010}\natexlab{}.
\newblock \showarticletitle{Improved algorithms for computing {F}isher's market
  clearing prices}. In \bibinfo{booktitle}{{\em STOC}}.
  \bibinfo{pages}{291--300}.
\newblock


\bibitem[\protect\citeauthoryear{Parkes, Procaccia, and Shah}{Parkes
  et~al\mbox{.}}{2012}]%
        {PPS12}
\bibfield{author}{\bibinfo{person}{D.C. Parkes}, \bibinfo{person}{A.D.
  Procaccia}, {and} \bibinfo{person}{N. Shah}.}
  \bibinfo{year}{2012}\natexlab{}.
\newblock \showarticletitle{Beyond dominant resource fairness: extensions,
  limitations, and indivisibilities}. In \bibinfo{booktitle}{{\em EC}}.
  \bibinfo{pages}{808--825}.
\newblock


\bibitem[\protect\citeauthoryear{Postlewaite and Schmeidler}{Postlewaite and
  Schmeidler}{1986}]%
        {PS}
\bibfield{author}{\bibinfo{person}{A. Postlewaite} {and} \bibinfo{person}{D.
  Schmeidler}.} \bibinfo{year}{1986}\natexlab{}.
\newblock \showarticletitle{Implementation in Differential Information
  Economies}.
\newblock \bibinfo{journal}{{\em JET\/}}  \bibinfo{volume}{39}
  (\bibinfo{year}{1986}), \bibinfo{pages}{14--33}.
\newblock


\bibitem[\protect\citeauthoryear{Procaccia}{Procaccia}{2013}]%
        {Pro13}
\bibfield{author}{\bibinfo{person}{A.D. Procaccia}.}
  \bibinfo{year}{2013}\natexlab{}.
\newblock \showarticletitle{Cake Cutting: Not Just Child's Play}.
\newblock \bibinfo{journal}{{\it Commun. ACM}} \bibinfo{volume}{56},
  \bibinfo{number}{7} (\bibinfo{year}{2013}), \bibinfo{pages}{78--87}.
\newblock


\bibitem[\protect\citeauthoryear{Procaccia and Wang}{Procaccia and
  Wang}{2014}]%
        {PW14}
\bibfield{author}{\bibinfo{person}{A.D. Procaccia} {and}
  \bibinfo{person}{Junxing Wang}.} \bibinfo{year}{2014}\natexlab{}.
\newblock \showarticletitle{Fair enough: guaranteeing approximate maximin
  shares}. In \bibinfo{booktitle}{{\em EC}}. \bibinfo{pages}{675--692}.
\newblock


\bibitem[\protect\citeauthoryear{Reny}{Reny}{1999}]%
        {Reny99}
\bibfield{author}{\bibinfo{person}{P.J. Reny}.}
  \bibinfo{year}{1999}\natexlab{}.
\newblock \showarticletitle{On the Existence of Pure and Mixed Strategy Nash
  Equilibria in Discontinuous Games}.
\newblock \bibinfo{journal}{{\em Econometrica\/}} \bibinfo{volume}{67},
  \bibinfo{number}{5} (\bibinfo{year}{1999}), \bibinfo{pages}{1029--1056}.
\newblock


\bibitem[\protect\citeauthoryear{Robertson and Webb}{Robertson and
  Webb}{1998}]%
        {RobertsonWebb98}
\bibfield{author}{\bibinfo{person}{J.M. Robertson} {and} \bibinfo{person}{W.A.
  Webb}.} \bibinfo{year}{1998}\natexlab{}.
\newblock \bibinfo{booktitle}{{\em Cake-cutting algorithms - be fair if you
  can}}.
\newblock \bibinfo{publisher}{A K Peters}. I--X, 1--181 pages.
\newblock
\showISBNx{978-1-56881-076-8}


\bibitem[\protect\citeauthoryear{Shapley and Shubik}{Shapley and
  Shubik}{1977}]%
        {ss-tp}
\bibfield{author}{\bibinfo{person}{L. Shapley} {and} \bibinfo{person}{M.
  Shubik}.} \bibinfo{year}{1977}\natexlab{}.
\newblock \showarticletitle{Trade using one commodity as a means of payment}.
\newblock \bibinfo{journal}{{\em Journal of Political Economy\/}}
  \bibinfo{volume}{85(5)} (\bibinfo{year}{1977}), \bibinfo{pages}{937--968}.
\newblock


\bibitem[\protect\citeauthoryear{Tullock}{Tullock}{1980}]%
        {Tullock80}
\bibfield{author}{\bibinfo{person}{G. Tullock}.}
  \bibinfo{year}{1980}\natexlab{}.
\newblock \showarticletitle{Efficient rent-seeking}.
\newblock In \bibinfo{booktitle}{{\em Toward a theory of the rent-seeking
  society}}, \bibfield{editor}{\bibinfo{person}{G.~Tullock J.M.~Buchanan,
  R.D.~Tollison}} (Ed.). \bibinfo{publisher}{Texas A. \& M, University Press}.
\newblock


\bibitem[\protect\citeauthoryear{Varian}{Varian}{1974}]%
        {Varian74}
\bibfield{author}{\bibinfo{person}{H.R. Varian}.}
  \bibinfo{year}{1974}\natexlab{}.
\newblock \showarticletitle{Equity, envy, and efficiency}.
\newblock \bibinfo{journal}{{\em Journal of Economic Theory\/}}
  \bibinfo{volume}{9}, \bibinfo{number}{1} (\bibinfo{year}{1974}),
  \bibinfo{pages}{63--91}.
\newblock


\bibitem[\protect\citeauthoryear{Young}{Young}{1995}]%
        {Young94}
\bibfield{author}{\bibinfo{person}{H.P. Young}.}
  \bibinfo{year}{1995}\natexlab{}.
\newblock \bibinfo{booktitle}{{\em Equity}}.
\newblock \bibinfo{publisher}{Princeton University Press}.
\newblock


\end{thebibliography}


\appendix

\section{Fisher Market: Perfect Complements}
\textsc{Theorem} \ref{thm:Fisher_Leontief_poa} (restated): 
\emph{
	The Fisher Market Mechanism with Leontief valuations has a price of anarchy of $n$ and the bound is tight.} 

\medskip
\begin{proof}
	Our tool is the following theorem, which states that Fisher markets with Leontief utilities have Nash equilibria 
	where players copy each others strategies.
	\begin{lemma}[\cite{BCDFFZ14}] \label{lem:uniform_leontief}
		The Fisher Market mechanism with Leontief preferences always has a Nash equilibrium
		where every buyer reports the uniform valuation $(1/m, \ldots, 1/m)$.
	\end{lemma}
	For completeness, we include the worst case example. Consider an instance with $n$ players of equal budgets ($\budget_i = 1$) and $n$ items, where each player $i$ likes item $i$ and nothing else; that is, $v_{i,i} = 1$, for all $i \in N$ and $v_{i,j} = 0, \forall i \in N, \forall j \neq i$. Then the optimal Nash Social Welfare is obtained in the Fisher market equilibrium, where the price of each item $j$ is $p_j = 1$ and the allocation is $x_{i,i} = 1, \forall i \in N$ and $x_{i,j}=0, \forall i \in N, \forall j \neq i$. 
	
	However, the strategy profile $\vec{y} = (\vec{y_1}, \ldots, \vec{y}_n)$, where $\vec{y}_i = (1/n, \ldots, 1/n)$ is a Nash equilibrium in which each player $i$ gets a fraction of $1/n$ from every item, yielding utility $1/n$ for every player. It follows that the price of anarchy is $n$.
	
	For a general upper bound, we note that any Nash equilibrium must be proportional, i.e. each player $i$ gets a fraction of at least $\budget_i/\totbudget$ of its best possible utility, $\mbox{OPT}_i$. Let $(\vec{p}, \vec{x})$ be a Nash equilibrium of the market, achieved under some reports $\vec{v}'$. Suppose for a contradiction that there exists player $i$ with $u_i(\vec{x}_i) < \budget_i/\totbudget \cdot \mbox{OPT}_i$. Then if $i$ reported instead its true valuation $\vec{v}_i$, the new market equilibrium, $(\vec{p}', \vec{x}')$, achieved under valuations $(\vec{v}_i, \vec{v}'_{-i})$, should satisfy the 
	inequality $u_i(\vec{x}'_{i}) \geq \budget_i/\totbudget \cdot \mbox{OPT}_i$. 
	If this were not the case, the outcome would not be a market 
	$(\vec{v}_i, \vec{v}_{-i}')$, since $i$ can afford the bundle $\vec{y} = (B_1/\totbudget, \ldots, B_m/\totbudget)$:
	\begin{equation} \label{eq:afford_weighted}
	\vec{p}(\vec{y}) = \sum_{j=1}^{m} p_j \cdot \frac{\budget_i}{\totbudget} = \frac{\budget_i}{\totbudget} \cdot \sum_{j=1}^{m} p_j = \budget_i,
	\end{equation}
	Moreover, $u_i(\vec{y}) = \budget_i/\totbudget \cdot \mbox{OPT}_i$, which together with Identity \ref{eq:afford_weighted} contradicts the market equilibrium property of $(\vec{p}', \vec{x}')$. 
	Thus in any Nash equilibrium $(\vec{p}, \vec{x})$ we have that $u_i(\vec{x}_i) \geq \budget_i/\totbudget \cdot \mbox{OPT}_i$, and so the Nash Social Welfare is $\mbox{NSW}(\vec{x}) \geq \prod_{i=1}^{n} \left(\frac{\budget_i}{\totbudget} \cdot \mbox{OPT}_i\right)^{\frac{\budget_i}{\totbudget}}$. Then the price of anarchy can be bounded as follows:
	\[
	\mbox{PoA} \leq \prod_{i=1}^{n} \frac{\mbox{OPT}_i^{\budget_i/\totbudget}}{\left( \frac{\budget_i}{\totbudget} \cdot \mbox{OPT}_i \right)^{\budget_i/\totbudget}} \leq \prod_{i=1}^{n} \left( \frac{\totbudget}{\budget_i}^{\budget_i/\totbudget} \right) \leq \frac{ \sum_{i=1}^{n} \budget_i \cdot \left( \frac{\totbudget}{\budget_i}\right)}{\totbudget} = n,
	\]
	where for the last inequality we used the fact that the weighted geometric mean is bounded by the weighted arithmetic mean (Lemma \ref{lem:budgets_bound}).
	This completes the proof of the theorem.
\end{proof}

\section{Trading Post: Perfect Complements}\label{app:tpc}

\textsc{Theorem \ref{thm:TP_char_exact}} (restated) : \emph{The Trading Post mechanism with Leontief utilities has exact pure Nash equilibria if and only if the corresponding Fisher market has market equilibrium prices that are strictly positive everywhere. When this happens, the Nash equilibrium utilities in Trading Post
	are unique and the price of anarchy is $1$.}
\begin{proof}
	Suppose $(\vec{v}, \vec{B})$ are valuations and budgets for which the Fisher market equilibrium prices are strictly positive everywhere (recall also that the valuations satisfy perfect competition). Then we show that some pure Nash equilibrium exists in the Trading Post game with the same valuations and budgets, and which induces the same allocation.
	
	Let $(\vec{p}, \vec{x})$ be the
	Fisher market equilibrium prices and allocation. Define matrix of bids $\vec{b}$ by $b_{i,j} = p_j \cdot x_{i,j}$, for all $i \in N, j \in M$. We claim that $\vec{b}$ is a pure Nash equilibrium in the Trading Post game. From the conditions in the theorem statement, we have that for each item $j$, there are two players $i \neq i'$ such that $b_{i,j} \cdot b_{i',j} > 0$. 
	Also the utility of each player $i$ in the market equilibrium is the same as that of strategy profile $\vec{b}$ and can be written as follows:
	\[
	u_i(\vec{x}_i) = \min_{j \in M: v_{i,j} > 0} \left\{ \frac{x_{i,j}}{v_{i,j}} \right\}
	\]
	By definition of the market equilibrium, each player $i$ gets each item in its demand set in the same fraction $f_i$, that is
	\[
	f_i = \frac{x_{i,j}}{v_{i,j}}, \mbox{for } j \in M \mbox{ with } v_{i,j} > 0
	\]
	Then player $i$ can only improve its utility by taking weight from some item(s) and shifting it towards others in its demand. However,
	since all the items are received in the same fraction $f_i$, it follows that player $i$ can only decrease its utility by such
	deviations. Thus the profile $\vec{b}$ is a pure Nash equilibrium of Trading Post.
	
	For the other direction, if a bid profile $\vec{b}$ is a pure Nash equilibrium in the Trading Post game, then consider the market allocation and prices $(\vec{x}, \vec{p})$, where $p_j = \sum_{k=1}^{n} b_{k,j}$ and the induced allocation $\vec{x}$ be given by $x_{i,j} = \frac{b_{i,j}}{\sum_{k=1}^{n} b_{k,j}}$. From the perfect competition requirement, clearly $x_{i,j}$ is always well defined. Then at $(\vec{x}, \vec{p})$ all the goods are allocated, all the money is spent, and each player gets an optimal bundle from its desired goods. To see the latter, note again that a player receives all the items in the same fractions at $(\vec{x}, \vec{p})$ and since all the prices are positive (since on each item there are at least two non-zero bids), then a player cannot decrease its spending on any item(s). Thus $(\vec{x}, \vec{p})$ is a market equilibrium and since the market equilibrium utilities are unique, it follows that the same is true for the PNE of Trading Post.
	
	From the correspondence between the Nash equilibria of Trading Post and the market equilibria of the Fisher mechanism, we obtain that on such instances the price of anarchy is $1$.
\end{proof}

\textsc{Theorem \ref{thm:TP_pne_existence}} (restated) : \emph{The parameterized Trading Post mechanism $\mathcal{TP}(\Delta)$ is guaranteed to have a pure Nash equilibrium for every $\Delta > 0$.}
\begin{proof}
	Let $\mathcal{TP}(\Delta)$ be the Trading Post game with minimum fee $\Delta$.
	We first show that a variant of the game, $\mathcal{TP}'(\Delta)$, where the strategy space of each player $i$ is restricted as follows, must have a pure Nash equilibrium.
	\begin{itemize}
		\item $i$ is forced to bid at least $\Delta$ on every item $j$ with the property that $v_{i,j} > 0$
		\item $i$ must bid zero on every item $j$ for which $v_{i,j} = 0$.
	\end{itemize}
	
	Clearly, the strategy space $S_i$ of each player $i$ is a nonempty compact convex subset of a Euclidean space. Moreover, the utility function of each player is continuous in $\vec{x}$ and quasi-concave in the player's own strategy (see, e.g. \cite{CM09}, chapter 2).
	
	We use the following theorem due to Debreu, Glicksberg, and Fan.
	
	\begin{lemma} (Debreu~\cite{Debreu52}; Glicksberg~\cite{Glicksberg52}; Fan 1952) \label{thm:Glicksberg}
		Consider a strategic form game whose strategy spaces $S_i$ are nonempty compact convex subsets of a Euclidean space. If the payoff functions $u_i$ are continuous in $s$ and quasi-concave in $s_i$, then there exists a pure strategy Nash equilibrium.
	\end{lemma}
	
	The conditions of Lemma~\ref{thm:Glicksberg} apply and so $\mathcal{TP}'(\Delta)$ has a pure Nash equilibrium $\vec{b}^{*}$.
	Consider now the Trading Post game with minimum bid $\Delta$, $\mathcal{TP}(\Delta)$. Note the strategy profile $\vec{b}^{*}$ dominates every other strategy in $\mathcal{TP}(\Delta)$, including those that allow the players to bid zero on items of interest to them, since such strategies can only decrease utility. Thus $\vec{b}^{*}$ is also a PNE in $\mathcal{TP}(\Delta)$, which completes the proof.
\end{proof}

\textsc{Theorem \ref{thm:apx_market}} (restated): \emph{The Nash social welfare at an $\epsilon$-market equilibrium for Leontief utilities is at least
	a $\frac{1}{(1+\epsilon)}$ factor of the optimal Nash social welfare.}
\begin{proof}
	For any given problem instance, let $(\pp',\xx')$ be an $\epsilon$-market equilibrium and let 
	$(\pp^*,\xx^*)$ be exact market equilibrium prices and allocation. By abuse of notation let $u_i(\pp')$ denote the optimal utility
	player $i$ can obtain at prices $\pp'$, {\em i.e.,} $u_i(\pp')= \max\left\{u_i(\yy) \ |\ \yy \ge 0;\ \ \pp' \cdot \yy \le B_i \right\}$. 
	
	
	For Leontief utility functions, convex formulation of (\ref{eq:eg}) captures the market equilibrium allocation. 
	%
	Note that, in order to get a utility of 1 at prices $\pp$, agent $i$ would need to spend 
	a total amount of money equal to $\phi_i(\pp)=\sum_j v_{ij}p_j$. \cite{dev} derived the dual of this convex program:
	
	\[
	\begin{array}{ll}
	\min: & \sum_j p_j - \sum_i B_i \log(\phi(\pp)) + \sum_i B_i \log(B_i) - \sum_i B_i\\
	s.t. & \forall j: p_j \ge 0
	\end{array}
	\]
	
	Note that the term $(\sum_i B_i \log(B_i) - \sum_i B_i)$ is a constant for a given market since $B_i$s are constants, and hence is
	omitted in \cite{dev}.
	Since $(\pp^*,\xx^*)$ is a market equilibrium, using strong duality and the fact that agents spend 
	all their money at equilibrium, i.e., $ \sum_j p^*_j = \sum_i B_i$: 
	
	\begin{equation}\label{eq.rt1}
	\sum_i B_i \log(u_i(\xx^*_i)) ~=~  - \sum_i B_i \log(\phi(\pp^*)) + \sum_i B_i \log(B_i) 
	\end{equation}
	
	Furthermore, at the $\epsilon$-market equilibrium $(\xx',\pp')$ all the agents spend all their money, 
	implying $\sum_j p'_j = \sum_i B_i$. 
	Since $\pp'$ is a feasible dual solution,
	
	
	
	\begin{equation*}
	- \sum_i B_i \log(\phi(\pp^*)) + \sum_i B_i \log(B_i)~\le~  - \sum_i B_i \log(\phi(\pp')) + \sum_i B_i \log(B_i).
	\end{equation*}
	Substituting the left hand side using Equation~\eqref{eq.rt1}, and taking an antilogarithm on both sides yields
	\begin{equation}\label{eq.rt2}
	\prod_i u_i(\xx^*_i)^{B_i} ~\le~ \prod_i \left(\frac{B_i}{\phi(\pp')}\right)^{B_i}.
	\end{equation}

	Since the optimal utility that agent $i$ gets at prices $\pp'$ is $u_i(\pp')$, which she derives using $B_i$ money, and while for unit
	utility she needs $\phi(\pp')$ money, we get
	
	\begin{equation}\label{eq.rt3}
	\forall i:\ \ u_i(\pp') = \frac{B_i}{\phi(\pp')} 
	\end{equation}
	
	Since $(\xx',\pp')$ is an $\epsilon$-market equilibrium, each agent gets an $\epsilon$-optimal
	bundle, so $u_i(\pp') \le u_i(\xx'_i)(1+\epsilon)$. According to \eqref{eq.rt3}, this implies
	$\frac{B_i}{\phi(\pp')}  ~\le~ u_i(\xx'_i)(1+\epsilon)$, which  combined with \eqref{eq.rt2} gives:
	\[
	\prod_i u_i(\xx^*_i)^{B_i} ~\le~ \prod_i \left(\frac{B_i}{\phi(\pp')}\right)^{m_i} ~\le~ (1+\epsilon)^{\totbudget} \prod_i u_i(\xx'_i)^{B_i}
	\]
	
	Since the Nash social welfare at $\xx$ is $(\Pi_i u_i(\xx_i)^{B_i})^{\frac{1}{\totbudget}}$, the result follows.
\end{proof}

\section{Trading Post: Beyond Perfect Substitutes and Complements}\label{app:tp_bsub}
In this section we prove upper bound of $2$ on PoA for Trading Post mechanism with arbitrary concave, non-decreasing valuation function and existence of pure Nash equilibria for CES functions. 
The proof for PoA bound is similar in structure to that of additive valuations.

\medskip

\textsc{Theorem \ref{thm:ub_tp_con}} (restated): \emph{The Trading Post Mechanism with concave valuations has price of anarchy at most 2.}
\begin{proof}
	Given a problem instance with concave valuations, let $\optx$ be the allocation that maximizes 
	the Nash social welfare subject to supply constraints. 
	Also, let $\eqx$ be the allocation obtained under a Nash equilibrium where each player $i$ bids $\eqbb_i$ and the price 
	of each item $j$ is $\eqp_j=\sum_i \eqbid_{i,j}$. 
	
	For any given player $i$, let $\devx$ be the allocation that arises if every player $k\neq i$ 
	bids $\eqbb_k$ while agent $i$ unilaterally deviates to $\vec{b}'_i$. 
	Think of this deviation as player $i$ first withdraw all its money, leaving price of item $j$ to be $p'_j = \eqp_j -\eqbid_{i,j}$, and then spending $\bid'_{i,j}$. Perfect competition condition ensures $p'_j>0,\ \forall j$, as at least two agents are interested in each good. Let the new bid $\bid'_{i,j}$ be such that for some $\beta_i>0$ and every item $j$.
	\begin{equation}\label{aeq:delta}
	\frac{\bid'_{i,j}}{p'_j+\bid'_{i,j}}~=~\frac{x^*_{i,j}}{\beta_i}.
	\end{equation}
	Bid $\vec{b}'_i$ is implied by the solution of the following program. 
	\[
	\text{min: } \beta_i\ \ \ \ \ \text{s.t.:}\ \ \ \  \frac{1}{\beta_i} = \frac{\bid'_{i,j}}{x^*_{i,j}(p'_j+\bid'_{i,j})}\ \ \mbox{ and }\ \ \bid'_{i,j}\ge 0, \ \ \ \forall j\in M;\ \ \ \ \ \ \ \sum_{j\in M} \bid'_{i,j} \le \budget_i
	\]
	
	Note that, $\bid'_{i,j}=0,\ \forall j$ and $\beta_i=\infty$ is a feasible point in the above program, and therefore it has a minimum.
	The allocation induced by this unilateral deviation of $i$ is 
	
	\[x'_{i,j}~=~\frac{\bid'_{i,j}}{p'_j+\bid'_{i,j}}~=~\frac{x^*_{i,j}}{\beta_i}.\]
	
	If $\beta_i\geq 1$, the utility of player $i$ after this deviation is $u_i\left(\frac{\optx_i}{\beta_i}\right)\ge u_i(\optx)/\beta_i$, 
	due to the concavity of $u_i$. If, on the other hand, $\beta_i <1$, then $u_i\left(\frac{\optx_i}{\beta_i}\right)\ge u_i(\optx)$ (non-decreasing). 
	Therefore, $u_i\left(\frac{\optx_i}{\beta_i}\right)\ge u_i(\optx)/\max\{1, \beta_i\}$. 
	But, $\eqx$ is a Nash equilibrium, so this deviation of player $i$ cannot yield a higher
	utility for $i$, which implies that 
	\begin{equation}\label{aineq:ratio_bound}
	u_i(\eqx)~\geq~ u_i\left(\frac{\optx_i}{\beta_i}\right)  ~~~~~ \Rightarrow ~~~~~ \frac{u_i(\optx)}{u_i(\eqx)}~\leq~ \max\{1, \beta_i\}.
	\end{equation}
	By definition of $\bid'$, we get $\sum_{j} \bid'_{i,j}=\sum_j x'_{i,j} (p'_j + \bid'_{i,j})=  \budget_i$; all the money is spent at minimum $\beta_i$. 
	Therefore, replacing for $x'_{i,j} =x^*_{i,j}/\beta_i$, we get 
	\begin{equation}\label{aeq:beta}
	\budget_i \beta_i = \sum_{j=1}^m x^*_{i,j} (p'_j + \bid'_{i,j})
	\end{equation}
	Also, we have $0\le x^*_{i,j}\leq 1$, and $p'_j \le \eqp_j$ for every item $j$. As a result, for each player $i$ we have,
	\begin{equation}\label{aineq:price_change}
	\sum_{j=1}^m x^*_{i,j}(p'_j+\bid'_{i,j}) ~\leq~ \sum_{j=1}^m (x^*_{i,j} \eqp_j)+\sum_{j = 1}^m (x^*_{i,j} \bid'_{i,j}) 
	~\leq~\sum_{j=1}^m \bid'_{i,j} +\sum_{j=1}^m x^*_{i,j}\eqp_j
	~\leq~ \budget_i + \sum_{j=1}^m x^*_{i,j}\eqp_j. 
	\end{equation}

	Let $N_{g}\subseteq N$ be the set of players for which $\beta_i \geq 1$, and $N_{\ell}$
	be the set of players with $\beta_i <1$. Using Inequalities~\eqref{aineq:ratio_bound},~\eqref{aeq:beta}, 
	and~\eqref{aineq:price_change}, and summing over all players gives 
	\begin{equation}\label{aeq:ratio}
	\sum_{i=1}^n \budget_i\frac{u_i(\vec{x}^*_i)}{u_i(\tilde{\vec{x}}_i)} ~\leq~ 
	\sum_{i=1}^n \budget_i \max\{1, \beta_i\} ~\leq~ 
	\sum_{i\in N_{\ell}} \budget_i  + \sum_{i\in N_g} \left(\budget_i + \sum_{j=1}^m x^*_{i,j}\eqp_j\right)  ~\leq~ 
	2\totbudget.
	\end{equation}
	Using Equation~\ref{aeq:ratio} and the inequality of weighted arithmetic and geometric means provided by Lemma~\ref{lem:budgets_bound},
	we conclude that 
	$$\left(\prod_{i\in N} \left(\frac{u_i(\vec{x}^*_i)}{u_i(\tilde{\vec{x}}_i)}\right)^{\budget_i}\right)^{1/\totbudget} ~\leq~ 2.$$
\end{proof}

\textsc{Theorem \ref{thm:existence_CES}} (restated): The Trading Post game with no minimum bid has exact pure Nash equilibria for all CES utilities with perfect competition and $\rho \in (-\infty, 1]$.
\begin{proof}
	The proof for $\rho = 0$ follows from the existence of Nash equilibria  in Fisher markets with Cobb-Douglas utilities \cite{BCDFFZ14}, for which the induced allocations are also proportional. Thus we assume $\rho \neq 0$.
	
	Our main tool for proving the existence of exact pure Nash equilibria is a result for discontinuous games due to \cite{Reny99}.
	First, we define the \emph{better-reply secure} property of a game with strategy space $S = S_1 \times \ldots \times S_n$ and utilities $u_i$.
	
	\begin{definition}
		Player $i$ can \emph{secure} a payoff of $\alpha \in \mathbb{R}$ at $s \in S$ if there exists $\bar{s_i} \in S_i$, such that $u_i(\bar{s_i}, s_{-i}^{'}) \geq \alpha$ for all $s_{-i}^{'}$ close enough to $s_{-i}$.
	\end{definition}
	
	\begin{definition}
		A game $G = (S_i, u_i)_{i=1}^{n}$ is \emph{better-reply secure} if whenever $(s^{*}, u^{*})$ is in the closure of the graph of its vector payoff function and $s^{*}$ is not a Nash equilibrium, then some player $i$ can secure a payoff strictly above $u_i^{*}$ at $s^{*}$.
	\end{definition}
	
	\begin{lemma}[Reny, 1999] \label{thm:Reny}
		If each $S_i$ is a nonempty, compact, convex subset of a metric space, and each $u_i(s_1, \ldots, s_n)$ is quasi-concave in $s_i$, then the game $G = (S_i, u_i)_{i=1}^{n}$ has at least one pure Nash equilibrium if in addition $G$ is better-reply secure.
	\end{lemma}
	
	The strategy spaces in Trading Post are nonempty, convex, compact subsets of the Euclidean space. Moreover the CES utility functions are quasi-concave for every parameter $\rho \leq 1$. We show the game is also better-reply secure.
	
	First note that Trading Post has discontinuities at strategy profiles where everyone bids zero on some item, captured by the set 
	$$\mathcal{D} = \{ \vec{b} \in S \; | \; \exists \; j \in M \; \mbox{such that} \; b_{i,j}=0, \forall i \in N\}.$$
	Let $(\vec{b}^{*}, \vec{u}^{*})$ be a point in the closure of the graph of the vector payoff function. If $\vec{b}^* \not \in \mathcal{D}$, then better-reply security holds at $(\vec{b}^{*}, \vec{u}^{*})$ by continuity. Thus we consider $\vec{b}^{*} \in \mathcal{D}$. Then there exists a sequence of strategy profiles $(\vec{b}^{K})_{K \geq 1}$ such that  
	$\vec{u}^{*} = \lim_{K \to \infty} (u_1(\vec{b}^{K}),$ $\ldots$, $u_n(\vec{b}^{K}))$, where $b^* = \lim_{K \to \infty} b^{K}$.
	
	Consider the set of items on which no player bids at $\vec{b}^{*}$:
	$J = \{j \in M \; | \; b_{i,j}^{*} = 0, \forall i \in N\}$ and $\ell$ an item in $J$. Denote by $C$ the set of players $i$ with $v_{i,\ell} > 0$ and $B_C = \sum_{i \in C} B_i$.
	By perfect competition, $|C| > 1$, and so there exists player $i$ and value $0 < \alpha < 1 - \frac{B_i}{B_C}$ such that $i$ gets a fraction of at most $f = \frac{B_i}{B_C} + \alpha < 1$ from item $\ell$ in every $K^{th}$ term of the sequence, possibly except the first $K_{\alpha}$ terms.
	
	Take the set of items
	that player $i$ declares valuable in the limit: $$L_i = \{ j \in M \; | \; b_{i,j}^* > 0\}.$$
	Given $0 < \epsilon < \min_{i,j} b_{i,j}^*$, construct an alternative profile $\vec{b}_i' (\epsilon)$ of player $i$, at which the bid on item $\ell$ is 
	$b_{i,\ell}'(\epsilon) = \epsilon$, the bid on items $j \in L_i$ is $b_{i,j}'(\epsilon) = b_{i,j}^* - \epsilon/|L_i|$, and the bid on items outside $L_i \cup \{\ell\}$ is $b_{i,j}'(\epsilon) = b_{i,j}^* = 0$.
	The strategy $\vec{b}_i'(\epsilon)$ guarantees that player $i$ gains the entire item $\ell$, while along the sequence $(\vec{b}^K)_{K \geq K_{\alpha}}$ the player was receiving a fraction of at most $f < 1$ from $\ell$. Moreover, by playing $b_{i,j}'(\epsilon)$, player $i$ loses a fraction of at most $\epsilon/w^*$ from every other item $j \in L_i$, where $w^* = \min_{j \in L_i} \sum_{k=1}^n b_{k,j}^*$. For CES utilities with $\rho > -\infty$, there exists small enough $\epsilon$ such that $u_i(\vec{b}_i'(\epsilon), \vec{b}_{-i}^*) > u_i^*$. Note that $u_i^* > u_i(\vec{b}^*)$, since player $i$ loses item $\ell$ at $\vec{b}^*$. 
	
	Moreover, by continuity of the utilities at $(\vec{b}_i'(\epsilon), \vec{b}_{-i}^*)$, the strategy $\vec{b}_i'(\epsilon)$ continues to guarantee player $i$ a payoff strictly above $u_i^*$ for any small enough perturbation of magnitude bounded by $\delta(\epsilon)$ of the strategies of the others around $\vec{b}_{-i}^*$. Thus the game is better-reply secure, which by Lemma \ref{thm:Reny} implies the existence of a pure Nash equilibrium.
\end{proof}

{\textsc{Remark.}
	\rn{We observe that the existence result may be extended to markets with strictly concave functions with some additional derivative conditions for specific directions. These conditions essentially ensure strict increase in the utility of an agent when she starts getting much more of a good, while the amounts of other goods she gets are decreased by close to zero quantity.}

	\medskip
	\medskip


\section{Beyond Pure Nash Equilibria in Trading Post}\label{app:TPP}
In this section we show that all Nash equilibria of the trading post game are pure. 
For this we will show that no matter what other players play, there is a {\em unique pure} best response strategy for player $i$. At Nash
equilibrium, since strategy of a player is a probability distribution over pure strategies that are in best response, this will imply
that every Nash equilibrium has to be pure. 

To show uniqueness of best response we will show that fixing mixed strategy/bid profile for all other players, player $i$'s utility is
a strictly concave function of its own bid profile. Let $\smi=(\ssigma_1,\dots,\ssigma_{i-1},\ssigma_{i+1},\dots,\ssigma_n)$ be a 
mixed-strategy profile of all the players except $i$, and let player $i$'s payoff function w.r.t. $\smi$ be denoted by
$u^{\smi}_i:S_i\rightarrow \Real$, where $S_i= \{(s_{i1},\dots, s_{im})\ge 0 \ |\ \sum_j s_{ij}= \budget_i\}$ is the set of pure
strategies (bids) of player $i$, {i.e.,} set of all possible ways in which she can split her budget across goods.

Next we derive the result for both additive and Leontief separately, and later extend it to concave under a mild assumption. Recall that we have assumed {\em perfect competition} where every good is liked by at least two players, which translates to for each good $j$, there exist $i \neq k \in N$ such that
$v_{i,j}\neq 0, v_{k,j}\neq 0$ for markets with additive or Leontief valuations. 

\subsection{Perfect Substitutes (Additive valuations)}
Function $u^\smi_i$ is as follows for the case of additive utilities.

\begin{equation}\label{eq:rlin}
u_i^\smi (\sss_i) = \sum_{\sss_{-i} \in S_{-i}} (\Pi_{k\neq i} \ssigma_i(\sss_k)) \sum_j v_{ij}
\frac{s_{ij}}{s_{ij}+\sum_{k \neq i} s_{kj}}, \ \forall \sss_i \in S_i
\end{equation}

Next we will show that function $u_i^\smi$ is strictly concave on the entire domain of $\Real^{m}$, therefore it is strictly
concave on $S_i$ as well. For this, we will show that Hessian of $u_i^\smi$ is negative definite.

\begin{lemma}\label{lem:sc}
	Given a mixed-strategy profile $\smi$ of all players $k\neq i$ such that every good is bought by at least one of them with non-zero
	probability, payoff function of player $i$, namely $u_i^\smi$, is strictly concave.
\end{lemma}
\begin{proof}
	We first show that $u_i^\smi$ is strictly concave on $\Rplus^m$.
	Take the derivative of $u_i^\smi$ with respect to $s_{ij}$ for each $j$,
	\[
	\frac{\partial u_i^\smi }{\partial s_{ij}} = \sum_{\sss_{-i} \in S_{-i}} (\Pi_{k\neq i} \ssigma_i(\sss_k)) v_{ij} \frac{\sum_{k\neq i}
		s_{kj}}{(s_{ij} + \sum_{k \neq i} s_{kj})^2} 
	\]
	
	Differentiating the above w.r.t. $s_{ig}$s for each good $g$, we get:
	
	\[
	\frac{\partial^2 u_i^\smi }{\partial s_{ij} \partial s_{ig}} = 0,\ \forall g \neq j;\ \ \ \ 
	\frac{\partial^2 u_i^\smi }{\partial s_{ij}^2} = -2 \sum_{\sss_{-i} \in S_{-i}} (\Pi_{k\neq i} \ssigma_i(\sss_k)) v_{ij}
	\frac{\sum_{k\neq i} s_{kj}}{(s_{ij}+\sum_{k \neq i} s_{kj})^3}
	\]
	
	Since bids are non-negative and at least one player other than $i$ is bidding on good $j$ with non-zero probability as per
	$\smi$, the second term in above expression is strictly negative. Therefore, the hessian is diagonal matrix with negative entries in
	diagonal, and hence negative definite. Thus, function $u^\smi_i$ is strictly concave on $\Real^m$. Therefore it remains strictly
	concave on convex subset $S_i\subset\Real^m$ as well.
\end{proof}

Strict concavity of $u^\smi_i$ established in Lemma \ref{lem:sc} implies that every player has a unique best response against any
strategy of the opponents and it is pure. At Nash equilibrium since only strategies in best response can have non-zero probability, it
follows that all Nash equilibria have to be pure. Thus we get the next theorem. 

\begin{theorem}\label{thm:linPureNE}
	In a market with linear utilities, every Nash equilibrium of the
	corresponding trading post game is pure. 
\end{theorem}
\begin{proof}
	At any given strategy profile, if only player $i$ is biding on a good $j$ with non-zero probability, then it can not be Nash equilibrium
	since that $i$ can reduce its bid on good $j$ to very small amount while still getting it fully, and use this saved money to buy other
	goods. 
\end{proof}
We observe that the equilibria of Trading Post with linear utilities are not necessarily unique.
\begin{example}\label{ex:lin}
	Consider a market with four buyers and two goods. Players 1 and 2 only want the first and second good, respectively, while players 
	3 and 4 like both goods equally. Then the bid profiles $\sss_1=(1,0), \sss_2(0,1), \sss_3=(1-\epsilon,\epsilon), \sss_4=(\epsilon,1-\epsilon)$ are in
	NE for any $0\le \epsilon \le 1$. 
\end{example}

\subsection{Perfect Complements (Leontief utilities)}
In case of Leontief utilities, the payoff of player $i$ is $\min_j \frac{x_{ij}}{v_{ij}}$, where $x_{ij}$ is the amount of good $j$
player $i$ gets. Like the linear case, for Leontief valuations too we will show uniqueness of best response, 
however the approach is different because the utility function $u_i^\smi$ (defined below) is
no more strictly concave. 

\begin{equation}\label{eq:rleon}
u_i^\smi(\sss_i) = \sum_{\sss_{-i} \in S_{-i}} (\Pi_{k\neq i} \ssigma_i(\sss_k)) \min_j\frac{1}{v_{ij}} 
\frac{s_{ij}}{s_{ij} +\sum_{k \neq i} s_{kj} },\ \ \forall \sss_i \in S_i
\end{equation}

We can show that $u_i^\smi$ is concave using the concavity of $\frac{s_{ij}}{s_{ij} + \sum_{k \neq i}
	s_{kj}}$ w.r.t. $s_{ij}$, however it is not strictly concave in general. 
Instead we will show that function $u_i^\smi$ has a unique optimum over all the possible strategies of player $i$ in game
$TP(\Delta)$. This will suffice to show that all equilibria are pure. 

Since every good is liked by at least two players, relevant bid profiles $\sss_{-i}$ are
only those where every good is bid on by some player $k$ other than $i$ since bidding zero fetches zero amount of the good. To capture
this we we define {\em valid} strategy profiles.

\begin{definition}
	Profile $\sss_{-i}$ is {\em valid} if $\forall j,\ \sum_{k\neq i} s_{kj} >0$. 
	Similarly, mixed-profiles $\smi$ is said to be {\em valid} if each pure profile $\sss_{-i}$ with $P(\sss_{-i})=\Pi_{k \neq i}
	\ssigma_k(\sss_k) >0$ is valid. 
\end{definition}

\begin{lemma}\label{lem:valid}
	If $\ssigma$ is a Nash equilibrium of game $TP(\Delta)$ for any $\Delta>0$, then profile $\smi$ is valid
	for each $i \in \CA$.
\end{lemma}
\begin{proof}
	For some $i$ if $\smi$ is not valid, then there exists $\sss_{-i}$ with $P(\sss_{-i})>0$ and good $j$ such that no one is bidding on it
	at $\sss_{-i}$. Let $k \neq i$ is an player with $v_{kj}>0$. Clearly she gets zero utility whenever she plays $\sss_k$. Instead if she 
	replaces $\sss_k$ with $\ttt$ where $t_j =\Delta$, $t_{j'} = s_{kj'} - \Delta$ where $s_{kj'}\ge 2\Delta$ (assuming $\Delta$ to be
	small there is such a good), and $t_d=s_{kd}, \forall d\neq j,j'$, will give her strictly better utility.
\end{proof}

Due to Lemma \ref{lem:valid} it suffice to consider only valid strategy profiles, both mixed as well as pure.

\begin{lemma}\label{lem:leonConcave}
	Given a valid strategy profile $\sss_{-i}$ of $TP(\Delta)$ for $\Delta>0$ and good $j$, consider function $f_j =
	\frac{1}{v_{ij}}\frac{s_{ij}}{s_{ij}+ \sum_{k \neq i} s_{kj}}$. $f_j$ seen as a
	function of $s_{ij}$ is strictly concave, and as
	a function of $(s_{i1},\dots,s_{in})$ it is concave.  \end{lemma}
\begin{proof}
	Taking double derivative of $f_j$ we get, 
	\[
	\begin{array}{c}
	\frac{\partial^2 f_j }{\partial s_{ij}^2} = \frac{-2}{v_{ij}} \frac{\sum_{k\neq i} s_{kj}}{(s_{ij}+\sum_{k \neq i} s_{kj})^3}\\
	\forall g \neq j,\ \frac{\partial^2 f_j }{\partial s_{ij} \partial s_{ig}} = 0 \ \ \ \mbox{ and } \frac{\partial^2 f_j }{\partial s_{ig}^2} = 0
	\end{array}
	\]
	First equality implies strict concavity w.r.t. $s_{ij}$ since $\sum_{k\neq i} s_{kj}\ge \Delta>0$ given that there is another player who 
	would want good $j$. Both equalities together imply that Hessian of $f_j$ is negative semi-definite and there by concavity w.r.t.
	$\sss_i$. 
\end{proof}

Next we show uniqueness of best response for player $i$ against any given $\smi$. For simplicity, we will abuse notation and use $u_i(\sss_i,\sss_{-i})$ to denote $u_i(\frac{s_{i,1}}{\sum_{k} s_{k,1}},\dots,\frac{s_{i,n}}{\sum_{k} s_{k,n}})$

\begin{lemma}\label{lem:leonUniqueOpt}
	Given a valid mixed-strategy profile $\smi$ of all players $k\neq i$ in trading post game $TP(\Delta)$ for any $\Delta>0$, payoff 
	function of player $i$, namely $u_i^\smi$, has a unique optimum. 
\end{lemma}
\begin{proof}
	First it is easy to see using Lemma \ref{lem:leonConcave} that each of the term inside summation of $u_i^\smi$ is a concave function
	w.r.t. $\sss_i$, since minimum of a set of concave functions is a concave function. And by the same argument $u_i^\smi$ itself is
	concave because it is summation of a set of concave functions. 
	
	To the contrary suppose there are two optimum $\sss_i, \sss'_i\in S_i$, $\sss_i\neq \sss'_i$. Clearly $\sum_j s_{ij}=\sum_j s'_{ij} =
	m_i$.  Therefore, there exists a good where bid at $\sss$ is more than that at $\sss'$. Let $j^*\in \{j\ |\ s_{ij}>s'_{ij}\}$. Clearly,
	$s_{ij^*}>\Delta$. 
	
	\begin{claim}
		Given $\sss$, if there is a good $j$ with $s_{ij}>\Delta>0$ and $u_i(\sss_i, \sss_{-i})< \frac{1}{v_{ij}}\frac{s_{ij}}{\sum_{k \neq i}
			s_{kj} + s_{ij}}$, then $\sss_i$ is not a best response to $\sss_{-i}$. 
	\end{claim}
	\begin{proof}
		It is easy to see that there exists $\delta$, such that for $t_{j}=s_{ij}-n\delta\ge \Delta$, $t_{k}=s_{ik}+\delta,\ \forall g\neq
		j$, we have $u_i(\sss) < u_i(\ttt,\sss_{-i})$. 
	\end{proof}
	
	$\forall \sss_{-i}$ with $P(\sss_{-i})>0$ if we have $u_i(\sss_i, \sss_{-i})< \frac{1}{v_{ij^*}}\frac{s_{ij^*}}{\sum_{k \neq i}
		s_{kj} + s_{ij^*}}$, then from the above claim it follows that $\sss_i$ is not an optimal solution of $u_i^\smi$. 
	
	Otherwise for some $\sss_{-i}$ we have $u_i(\sss_i,\sss_{-i})=\frac{1}{v_{ij^*}}\frac{s_{ij^*}}{\sum_{k \neq i}
		s_{kj} + s_{ij^*}}$. Due to concavity of $u^\smi_i$ we have that entire line-segment from $\sss_i$ to $\sss'_i$ is optimal. Call this
	line segment $\CL$. On this line-segment bid on good $j$ is strictly changing (decreasing). Since $u_i(\sss_i,\sss_{-i})=\min_j
	\frac{1}{v_{ij}}\frac{s_{ij}}{\sum_{k \neq i}
		s_{kj} + s_{ij}}$ is governed by the utility obtained from good $j^*$, it is strictly concave on $\CL$ at $\sss_i$. 
	
	Furthermore, at a point if a set of functions are concave and one of them is strictly concave then their summation is strictly concave.
	Therefore, $u_i^\smi$ is strictly concave at $\sss_i$ on $\CL$ and hence either $\sss_i$ is not optimum or other points on $\CL$ are not
	optimum. In either case we get a contradiction.
\end{proof}

Since at Nash equilibrium only optimal strategies can have non-zero probability, the next theorem follows using Lemmas \ref{lem:valid}
and \ref{lem:leonUniqueOpt}.

\begin{theorem}\label{thm:leoPureNE}
	For market with Leontief utilities, 
	every Nash equilibrium of the corresponding $\Delta>0$ trading post game $TP(\Delta)$ is pure. 
\end{theorem}

The Trading Post game with Leontief valuations does not always have unique PNE, as can be seen from the next example.

\begin{example}\label{ex:leo}
	Let there be an instance of Trading Post with players $N = \{1, 2\}$, items $M = \{1,2\}$, and Leontief valuations $v_{i,j} = 1$, $\forall i,j
	\in \{1, 2\}$.  Then every strategy profile of the form $v_{1,1} = v_{2,1} = a$, with $a \in (0,1)$ is a pure Nash equilibrium of
	Trading Post, since both players get the items in the optimal ratios. 
\end{example}

\subsection{Beyond Perfect Substitutes and Complements: Concave Valuations}
As stated in preliminaries section, valuation function of each player $i$ namely $u_i:\Real_+^m
\rightarrow \Real$ is concave, non-negative, and non-decreasing in general, where $m=|M|$ is the number of goods in the market. 
Atypical assumption on $u_i$ is that it is {\em non-satiated} \cite{AD54}, {\em i.e.,} $u_i(\xx) > u_i(\yy)$ if $x_j\ge y_j,\ \forall j$ and $x_j>y_j$ for some $j$. We assume following stronger notion of perfect competition.

\medskip

\noindent{\bf Enough Competition.} For every good $j$, there are two players $k\neq i$, such that $\frac{\partial u_i}{\partial x_{ij}}(\xx)$ and $\frac{\partial u_k}{\partial x_{kj}}(\yy)$ are infinity at bundles $\xx,\yy \in \Real^m_+$ if $x_j=0$ and $y_j=0$ respectively. 

\medskip

In Trading Post mechanism, since strategy of a player is to specify money bid on each good, for $\sss_i\in S_i$ and $\sss_{-i}\in S_{-i}$ let us define
\begin{equation}\label{eq:us}
v_i(s_i,s_{-i}) = v_i\left(\frac{s_{i,1}}{s_{i,1}+\sum_{k \neq i} s_{k,1}},\dots,\frac{s_{i,m}}{s_{i,m}+ \sum_{k \neq i} s_{k,m}}\right)
\end{equation}

Then, function $u^\smi_i$ can be written as follows. 

\begin{equation}\label{eq:rcon}
u_i^\smi (\sss_i) = \sum_{\sss_{-i} \in S_{-i}} (\Pi_{k\neq i} \ssigma_i(\sss_k)) v_i(\sss_i,\sss_{-i}),
\ \forall \sss_i \in S_i
\end{equation}

Next we will show that function $u_i^\smi$ is strictly concave on the entire domain of $\Real^{m}$ given that $\smi$ is a NE, therefore it is strictly
concave on $S_i$ as well. For this, first we show that $u_i$ is (strictly) concave in $\sss_i$.

\begin{lemma}\label{lem:ui_sconv}
	For a fixed $\sss_{-i}\in S_{-i}$, function $v_i$ is concave in $\sss_i$, and is strictly concave if $\sss_{-i}$ satisfies $\sum_{k \neq i} s_{ij} >0, \forall j \in M$. 
\end{lemma}
\begin{proof}
	We will show a general claim, and then apply it to our setting. Let $[m]$ denote the set $\{1,\dots,m\}$.
	\begin{claim}
		Given concave, non-decreasing, non-satiated function $f:\Real^m_+\rightarrow \Real$, and $m$ single variate concave functions $g_i:\Real_+\rightarrow \Real_+$ for $i\in[m]$, composition function $f(g_1(x_1),\dots,g_m(x_m))$ is concave. Further if either $f$ or all $g_i$'s are strictly concave then $f(g_1(x_1),\dots,g_m(x_m))$ is strictly concave.
	\end{claim}
	\begin{proof}
		Let $\xx,\yy \in \Real^m_+$, and $\xx'$ and $\yy'$ be such that $x'_i = g_i(x_i)$ and $y'_i=g_i(y_i)$ for all $i \in [m]$. 
		Let $\lambda \in (0, 1)$ be a constant and $\zz=\lambda \xx + (1-\lambda) \yy$, then by concavity of $g_i$'s we get $g_i(z_i)\ge \lambda x'_i + (1-\lambda)y'_i,\forall i\in[m]$. Using this we get, 
		\[
		\begin{array}{l}
		\lambda f(g_1(x_1),\dots,g_m(x_m)) +(1-\lambda) f(g_1(y_1),\dots,g_m(y_m)) \\
		\ \ \ \ \ \  \ \ \ \ \ \ \ \ \ \ \ \ \begin{array}{cl}
		=& \lambda f(\xx') + (1-\lambda) f(\yy') \\
		\le & f(\lambda \xx' + (1-\lambda) \yy')  \ \ \ \ \ \ \ \ (\because \mbox{ $f$ is concave})\\
		\le & f(g_1(z_1),\dots,g_m(z_m))\ \ \ (\because \mbox{ $g_i$s are concave and $f$ is non-decreasing}) \\
		\end{array}
		\end{array}
		\]
		Above, second inequality become strict if $f$ is strictly concave, and the third inequality becomes strict if $g_i$'s are strictly concave due to $f$ being non-satiated.  
	\end{proof}
	
	For our purpose, set $f=v_i$ which is concave, non-decreasing, and non-satiated. Given $\sss_{-i}$ let $D_j = \sum_{k \neq i} s_{ij}$ and set $g_j(s_{ij})=\frac{s_{ij}}{s_{ij} + D_j},\ \forall j\in M$. Clearly, where $\frac{0}{0}$ is considered as $0$.
	\[
	\frac{\partial g_j}{\partial s_{ij}} = \frac{D_j}{(s_{ij}+D_j)^2}\ge 0,\ \ \ \ \frac{\partial^2 g_j}{\partial s_{ij}^2} = \frac{-2D_j}{(s_{ij}+D_j)^3} \le 0
	\]
	
	Thus, $g_j$ is concave, and is strictly concave if $D_j>0$. Since, $v_i(.,\sss_i) = f(g_1(s_{i1}),\dots,g_m(s_{im}))$, proof follows using the above claim.
\end{proof}

Using the property of $v_i$ established in the above lemma, next we will show that Hessian of $u_i^\smi$ is negative definite almost always.

\begin{lemma}\label{lem:sc}
	Given a mixed-strategy profile $\smi$ of all players $k\neq i$ such that there is an $\sss'_{-i}\in S_{-i}$ played with positive probability where $\sum_{k\neq i} s'_{k,j} > 0,\ \forall j$, then payoff function of players $i$, namely $u_i^\smi$, is strictly concave.
\end{lemma}
\begin{proof}
	Fix any $\sss_i \in S_i$, and let $H$ be the Hessian of $u_i^\smi$ at $\sss_i$. By definition of strictly concave functions it suffices to show that $H$ is negative definite. Let $H(\sss_{-i})$ be the Hessian of $v_i(.,\sss_{-i})$ at $\sss_i$ for each $\sss_{-i} \in S_{-i}$. Let $\alpha(\sss_i)=(\Pi_{k\neq i} \ssigma_i(\sss_k))$, then by definition of $u_i^{\smi}$ from (\ref{eq:rcon}) Hessian of $u_i^\smi$ at $\sss_i$ is $H=\sum_{\sss_{-i} \in S_{-i}} \alpha(\sss_i) H(\sss_{-i})$.

	By Lemma \ref{lem:ui_sconv} since each of these $v_i(.,\sss_{-i})$ is concave, $H(\sss_{-i})$ is negative semi-definite. Furthermore, the one with $\sum_{k\neq i} s'_{k,j}>0,\ \forall j$ and $\alpha(\sss'_{-i})>0$ due to the hypothesis gives strictly concave $v_i(.,\sss'_{-i})$ (Lemma \ref{lem:ui_sconv}), and hence $H(\sss'_{-i})$ is negative definite. For any $\xx \in \Real^m$, we have 
	\[
	\xx^T H \xx = \xx^T \left(\sum_{\sss_{-i} \in S_{-i}} \alpha(\sss_i) H(\sss_{-i})\right) \xx = \sum_{\sss_{-i} \in S_{-i}} \alpha(\sss_i) (\xx^TH(\sss_{-i})\xx) <0
	\]
	
	The last inequality follows from the fact that $\xx^TH(\sss_{-i})\xx \le 0,\ \forall \sss_{-i}\in S_{-i}$, and $\xx^TH(\sss'_{-i})\xx<0$. By definition of negative definite matrices, the proof follows. 
\end{proof}

Using {\em enough competition} condition and strict concavity of $u^\smi_i$ established in Lemma \ref{lem:sc}, next we show the main result. 

\begin{theorem}\label{thm:conPureNE}
	In a market with concave utilities and enough competition, every Nash equilibrium of the
	corresponding trading post game is pure. 
\end{theorem}
\begin{proof}
	Let $\ssigma=(\ssigma_1,\dots,\ssigma_n)$ be a Nash equilibrium profile, and to the contrary suppose it is mixed. Let $i$ be an player who is playing a mixed strategy. At Nash equilibrium since only strategies in best response can have non-zero probability, there are $\sss_i \neq \sss'_i \in S_i$ such that $u_i^\smi$ is maximized at both. 
	
	Due to enough competition for every good $j$ there is an player $k\neq i$ who always bids on good $j$, i.e., $\forall \sss_k\in S_k$ with $\ssigma_k(\sss_k)>0$, we have $s_{kj}>0$. Therefore, at $\smi$ every $\sss_{-i}$ with non-zero probability satisfies $\sum_{k\neq i} s_{kj}>0,\ \forall j$. Using this Lemma \ref{lem:ui_sconv} implies $u_i^\smi$ is strictly concave and there by has a unique maximum. A contradiction to both $\sss_i$ and $\sss'_i$ being maximizer of $u_i^\smi$. 
\end{proof}

\section{Fairness Guarantees} 

\textsc{Theorem} \ref{thm:TP_fairness} (restated): 
\emph{Every Nash equilibrium of the 
	\begin{itemize}
		\item Fisher market mechanism with concave utilities
		\item Trading Post mechanism with concave and strictly increasing utilities 
	\end{itemize}
	is weighted proportional. Moreover, for every $0 < \Delta < 1/m$, every Nash equilibrium of the parameterized Trading Post game $TP(\Delta)$ with concave and strictly increasing utilities is approximately proportional, guaranteeing at least $\frac{B_i}{\mathcal{B}}\left(1 - \rho_i\right)$ of the optimum for each player $i$, where $\rho_i = \frac{\Delta \cdot(m-1)}{B_i}$, $B_i$ is the player's budget, $\mathcal{B}$ the sum of budgets, and $m$ the number of items.}
\begin{proof}
	We consider first the Fisher market. Let $\vec{u}$ be the true utilities of the players and $\vec{u}^*$ a pure Nash equilibrium of the market. Given a bundle $\vec{x}$, 
	denote by $u_i(\vec{x})$ the (true) utility of player $i$ when receiving $\vec{x}$. Given utilities $\vec{w}$, denote by $u_i(\vec{w})$ the (true) utility of player $i$ for the Fisher market equilibrium allocation obtained when the players declare $\vec{w}$ as their utilities.
	
	A strategy that is always available to player $i$ when facing $\vec{u}^*_{-i}$ is its true valuation, $\vec{u}_i$. Since $\vec{u}^*$ is a pure Nash equilibrium, $u_i(\vec{u}_i^*, \vec{u}_{-i}^*) \geq u_i(\vec{u}_i, \vec{u}_{-i}^*)$. Moreover, player $i$ can always afford to purchase the bundle $\frac{B_i}{\mathcal{B}} \cdot \vec{1}$, which consists of a fraction of $\frac{B_i}{\mathcal{B}}$ of each good $j$. When $i$ is honest, the market equilibrium allocation must give a bundle worth at least as much as
	$\frac{B_i}{\mathcal{B}} \cdot \vec{1}$, and so $u_i(\vec{u}_i, \vec{u}_{-i}^*) \geq u_i(\frac{B_i}{\mathcal{B}} \cdot \vec{1})$. Since the utilities are concave and non-negative, we have that 
	$u_i(\frac{B_i}{\mathcal{B}} \cdot \vec{1}) \geq \frac{B_i}{\mathcal{B}} \cdot u_i(\vec{1})$. By combining the previous inequalities, we get
	$$
	u_i(\vec{u}^*) \geq u_i(\vec{u}_i, \vec{u}_{-i}^*) \geq
	u_i \left(\frac{B_i}{\mathcal{B}} \cdot \vec{1} \right)
	\geq \frac{B_i}{\mathcal{B}} \cdot u_i(\vec{1})
	$$
	
	
	For Trading Post, 
	we first show that the PNE of the game with no minimum bid achieve weighted proportionality when the utilities are concave and strictly increasing. Note that for valuations such as Leontief, Nash equilibria exist only when the corresponding Fisher market prices are strictly positive; we handle these cases later through the parameterized game.
	
	Recall that for Trading Post we normalize the budgets such that $B_i \geq 1$. Let $\eqbb$ be a PNE strategy profile. For any player $i$, let $D_j = \sum_{k\neq i} \tilde{b}_{k,j}$ be the sum of bids from the other players at good $j$.
	We wish to argue that player $i$ has a ``safe'' strategy that guarantees it a fraction of $B_i/\mathcal{B}$ from each good.
	Define strategy $\vec{y} = (y_1, \ldots, y_m)$ of player $i$, such that
	$y_j = \frac{B_i \cdot D_j}{\mathcal{B} - B_i}.$ for each item $j$. Since the utilities are strictly monotonic, it follows that in any PNE we have $D_j > 0$ (since otherwise a player could gain by halving its bid on item $j$, still win the item, and increase the fractions from other items),
	and so $y_j > 0$ for all $j$. It can be verified that $\sum_{j=1}^{m} y_j  = B_i$ and that 
	the fraction received by player $i$ from every good $j$ is:
	\[
	\frac{y_j}{y_j + D_j} = \frac{\left(\frac{B_i \cdot D_j}{\mathcal{B} - B_i}\right)}{\left(\frac{B_i \cdot D_j}{\mathcal{B} - B_i}\right) + D_j} = \frac{B_i \cdot D_j}{B_i \cdot D_j + D_j \cdot (\mathcal{B} - B_i)} = \frac{B_i}{\mathcal{B}}.
	\]
	Thus $\vec{y}$ is a safe strategy that guarantees player $i$ a fraction of $B_i/\mathcal{B}$ of every good. Since $\eqbb$ is a PNE, strategy $\vec{y}$ is not an improvement, which together with concavity of the utilities gives: 
	$$u_i(\eqbb) \geq u_i(\vec{y},\tilde{b}_{-i}) \geq 
	u_i \left(\frac{B_i}{\mathcal{B}} \cdot \vec{1}\right) \geq \frac{B_i}{\mathcal{B}} \cdot u_i(\vec{1}).$$
	
	Consider now the parameterized Trading Post game. For any $0 < \Delta < 1/m$, let $\eqbb$ be a PNE strategy profile and the quantities $D_j$ and $\vec{y}$ as defined above. Let $S_i = \{j \in [m] \; | \; y_j < \Delta\}$ be the set of items on which the safe strategy $\vec{y}$ is below the minimum bid. Note it cannot be the case that $|S_i| = m$ since $\Delta < 1/m$, $\sum_j y_j = B_i$, and $B_i \geq 1$, thus $|S_i| \leq m-1$. Let $B_i' = B_i - \Delta \cdot |S_i|$ and define a modified strategy profile $\vec{z} = (z_1, \ldots, z_m)$ for player $i$:
	\[
	z_j =\left\{
	\begin{array}{ll}
	\Delta, & \mbox{if} \; j \in S_i \\                
	\frac{B_i' \cdot D_j}{\mathcal{B} - B_i'},  & \mbox{otherwise}
	\end{array}
	\right.
	\]
	We show first that the strategy $\vec{z}$ is feasible, by not exceeding $i$'s budget:
	\begin{eqnarray*}
		\sum_{j=1}^{m} z_j & = & \Delta \cdot |S_i| + \sum_{j \not \in S_i} z_j = \Delta \cdot |S_i| + \sum_{j \not \in S_i} \frac{B_i' \cdot D_j}{\mathcal{B} - B_i'} = \Delta \cdot |S_i| + \left( \frac{B_i - \Delta \cdot |S_i|}{\mathcal{B} - B_i + \Delta \cdot |S_i|} \right) \cdot \left( \sum_{j \not \in S_i} D_j \right) \\
		& \leq & \Delta \cdot |S_i| + \left( \frac{B_i - \Delta \cdot |S_i|}{\mathcal{B} - B_i + \Delta \cdot |S_i|} \right) \cdot (\mathcal{B} - B_i) \leq B_i
	\end{eqnarray*}
	Clearly for each item $j \in S_i$, player $i$ receives a fraction of at least $B_i/\mathcal{B}$---this fraction was guaranteed at the bid $y_j$, and $z_j = \Delta > y_j$. From the items $j \not \in S_i$, player $i$ gets a fraction of:
	\[
	\frac{z_j}{z_j + D_j} = \frac{\left( \frac{B_i' \cdot D_j}{\mathcal{B} - B_i'}\right)}{\left( \frac{B_i' \cdot D_j}{\mathcal{B} - B_i'}\right) + D_j} = \frac{B_i'}{\mathcal{B}} = \frac{B_i - \Delta \cdot |S_i|}{\mathcal{B}} \geq \frac{B_i - \Delta \cdot (m-1)}{\mathcal{B}} = \frac{B_i}{\mathcal{B}} \left(1 - \frac{\Delta \cdot (m-1)}{B_i}\right)
	\]
	Let $\rho_i = \frac{\Delta \cdot (m-1)}{B_i}$.
	By concavity and strict monotonicity, we have:
	$$
	u_i(\eqbb) \geq u_i(\vec{z}, \tilde{b}_{-i}) \geq u_i\left( \frac{B_i}{\mathcal{B}} \left(1 - \rho_i\right) \cdot \vec{1}\right) \geq  \frac{B_i}{\mathcal{B}} \left(1 - \rho_i\right) \cdot u_i(\vec{1})
	$$
	Also note the fraction converges to $B_i/\mathcal{B}$ as $\Delta \to 0$. 
	This completes the proof.
\end{proof}


\end{document}